\documentclass[11pt]{article}

\usepackage[margin=1in]{geometry}
\usepackage{setspace}
\usepackage{natbib}
\usepackage{appendix}
\usepackage{dsfont}
\usepackage{amsmath,amssymb,amsfonts,amsthm,mathtools}
\usepackage{bm}
\usepackage{bbm}
\usepackage{tikz}
\usetikzlibrary{arrows.meta, positioning}
\usepackage{comment}
\usepackage{algorithm}
\usepackage{algpseudocode}
\usepackage{graphicx}
\usepackage{booktabs}
\usepackage{multirow}
\DeclareMathOperator*{\argmax}{arg\,max}

\newcommand{\kl}{D_{\operatorname{KL}}}
\usepackage{parskip}
\usepackage{authblk}

\usepackage[colorlinks=true,
            linkcolor=blue,
            citecolor=blue,
            urlcolor=blue]{hyperref}

\usepackage{cleveref}

\newcommand{\X}{X}
\newcommand{\x}{x}
\newcommand{\Y}{Y}

\newcommand{\D}{D}
\newcommand{\Q}{Q}

\newcommand{\PP}{P}

\newcommand{\cX}{{\cal \X}}
\newcommand{\cY}{{\cal \Y}}

\newcommand{\prob}{\mathbb P}

\theoremstyle{plain}
\theoremstyle{plain}
\newtheorem{theorem}{Theorem}[section]
\newtheorem{proposition}[theorem]{Proposition}

\theoremstyle{definition}
\newtheorem{definition}[theorem]{Definition}

\theoremstyle{remark}

\definecolor{codegreen}{rgb}{0,0.6,0}
\definecolor{codegray}{rgb}{0.5,0.5,0.5}
\definecolor{codepurple}{rgb}{0.58,0,0.82}
\definecolor{backcolour}{rgb}{0.95,0.95,0.92}

\crefname{assumption}{assumption}{assumptions}
\Crefname{assumption}{Assumption}{Assumptions}

\crefname{theorem}{theorem}{theorems}
\Crefname{theorem}{Theorem}{Theorems}

\crefname{lemma}{lemma}{lemmas}
\Crefname{lemma}{Lemma}{Lemmas}

\crefname{figure}{fig.}{}
\Crefname{figure}{Fig.}{}

\title{
\bfseries Optimal e-variables under constraints
}

\author[1]{Aytijhya Saha}
\author[2]{Aaditya Ramdas}

\affil[1]{Massachusetts Institute of Technology.  \texttt{aytijhya@mit.edu}}
\affil[2]{Carnegie Mellon University. \texttt{aramdas@cmu.edu}}
\begin{document}
\date{}

\maketitle

\begin{abstract}
E-variables enable safe and anytime-valid inference, with log-optimal
e-variables given by the likelihood ratio of the 
least favorable distributions (LFDs) when they exist in composite settings.
While this unconstrained theory is well understood, one may need/wish to impose additional structural constraints, including differential privacy, quantization, boundedness, or moment restrictions.
We show that under these constraints,
log-optimal constrained e-variables can often be constructed by a simple \emph{optimize-then-constrain} principle:
first compute the unconstrained log-optimal e-variable, then impose the constraint via an appropriate transformation.
Thus, the constrained growth-rate optimization problem
does not require solving for a different LFD pair; the constrained optimal solution is just a post-processing of the unconstrained optimal solution.
\end{abstract}

\section{Introduction}
E-variables provide a foundation for safe and anytime-valid inference \citep{ramdas2022game,ramdas2024hypothesis}.
A nonnegative random variable $E$ is an e-variable for a null
$H_0$ (set of distributions consistent with a null hypothesis) if
\begin{equation}
\label{def-e-val}
    \sup_{P_0 \in H_0} \mathbb{E}_{P_0}[E] \le 1.
\end{equation}
The value realized by an e-variable (that is, its instantiation) is called an e-value.
In simple versus simple hypothesis testing, the likelihood ratio maximizes expected log-growth under the alternative \citep{shafer2021testing,kelly1956new,breiman1961optimal}.
In composite versus composite testing, \cite{saha2025huber} shows that log-optimal or growth-rate-optimal in the worst case (GROW) e-variables \citep{grunwald2020safe} arise from so-called ``least favorable
distributions'' (LFDs) \citep{huber1965robust}, when they exist.

While this unconstrained theory is well understood, one may often need to impose additional structural constraints, including local differential privacy, quantization, boundedness, or moment restrictions.
Each constraint defines a restricted class of admissible e-variables.
This raises a fundamental question:
\emph{How does one construct growth-rate optimal e-variables for composite hypothesis testing problems 
under structural constraints?}
At first glance, the answer appears problem-specific.
One might expect that each constraint requires solving
a new problem,
ignoring the unconstrained optimal solution. This paper shows that, for various constraints,
the situation is far more structured.
Our central message is:
\emph{Constrained growth-rate optimal e-variables
are often obtained by first solving the unconstrained problem,
and then later imposing the constraint via an appropriate (deterministic or random) transformation.}

In this paper, we demonstrate this ``optimize-then-constrain'' phenomenon through the following four different types of constraints and provide log-optimal/GROW e-variables for composite testing problems in the presence of an LFD pair.

\begin{enumerate}
\item \textbf{Local differential privacy (LDP):} 
 LDP has become a key model
for privacy-preserving inference
\citep{duchi2023local,kairouz2016extremal}.
Under $\varepsilon$-LDP,
the observable output $Y$ is generated from $X$
through a randomized channel $Q(\cdot \mid X)$ satisfying
\[
\sup_{x,x'} \sup_{y}
\frac{Q(y \mid x)}{Q(y \mid x')} \le e^{\varepsilon}.
\]
The e-variable cannot depend directly on $X$,
but only on its privatized version $Y$.
\item \textbf{Quantization.} In distributed or decentralized systems \citep{tsitsiklis1989decentralized,veeravalli1993decentralized}, local nodes cannot transmit real-valued statistics with arbitrary precision. Instead, they can send only a finite number of bits.
In such settings, the e-variable must take values in a finite set. Mathematically, this forces the e-variable to be approximated by a step function.
\item \textbf{Boundedness.} In sequential decision-making, unbounded likelihood ratios can create instability. Extremely large e-variables may dominate aggregation rules, lead to numerical instability,
or amplify rare but extreme events.
Bounding e-variables is also essential 
in robustness
and global privacy-preserving mechanisms that limit sensitivity
\citep{dwork2014algorithmic}.  
Imposing e-variables to be bounded in some interval $[c_1,c_2]$ might be crucial in such settings.
\item \textbf{Bounded convex integral constraints.}
Another related constraint is to impose an upper bound on the variance
(or higher moments) under the null or alternative. We formalize this mathematically through a broad class of convex integral constraints.
\end{enumerate}
In all these settings, we will show that our ``optimize-then-constrain'' framework applies.

\subsection{Background and related work}
Recall the definition of e-variable from \eqref{def-e-val}. E-variables measure evidence against the null: the larger its realized value, the stronger
the evidence. One can also use e-variables to make hard accept/reject decisions. In particular, Markov’s inequality implies that we can reject $H_0$ at level $\alpha$ if an e-variable $E$ exceeds $1/\alpha$,
since $P[E\geq 1/\alpha]\leq \alpha$ for each $P\in H_0$. In this paper, however, our primary focus is rather on the optimality of $E$ itself, 
a level-$\alpha$ test can always be obtained by thresholding $E$ at $1/\alpha$.

In the simple-versus-simple setting $H_0=\{P_0\}$ and $H_1=\{P_1\}$, 
suppose $P_0$ and $P_1$ admit densities $p_0$ and $p_1$ with respect to a common dominating measure. 
Then the likelihood ratio $\frac{p_1(X)}{p_0(X)}$ is the log-optimal e-variable, that is, $ \mathbb E_{P_1}\left[\log\frac{p_1(X)}{p_0(X)}\right]\geq\mathbb E_{P_1}(\log B),$ for any e-variable $B$ under $P_0$
\citep{shafer2021testing}, a result with roots in \cite{kelly1956new,breiman1961optimal}. 
When the null and alternative are composite,
$H_0 = \mathcal{P}_0$ vs.\ $H_1 = \mathcal{P}_1,$
log-optimality is defined in a minimax sense. \citet{grunwald2020safe} introduced the notion of \emph{growth-rate optimal in worst case} (GROW) e-variables, which achieves the worst-case expected log-growth over $\mathcal{P}_1$ while maintaining validity over $\mathcal{P}_0$:
\[\sup_{E\in\mathcal{E}(\mathcal{P}_0)}\inf_{P_1\in\mathcal{P}_1}\mathbb{E}_{P_1}[\log E],\]
where $\mathcal{E}(\mathcal{P}_0)$ denotes the set of all e-values under $\mathcal{P}_0$.
A key concept in this setting is that of a least favorable distribution (LFD) pair, which we define next.
Following \citet{huber1965robust}, consider a test $\phi$ between $\mathcal{P}_0$ and $\mathcal{P}_1$, and define the  risk
\begin{align}
    \label{risk}
\nonumber    &  R(P_0, \phi) = C_0 P_{0}(\phi \text{ accepts } \mathcal{P}_1),  \\      ~~
       & R(P_1, \phi) = C_1 P_{1}(\phi \text{ accepts } \mathcal{P}_0),
\end{align}
where $C_0, C_1 > 0$ are fixed constants.
The formal definition of a least favorable distribution pair is given below.  
\begin{definition}
$(P_0^*, P_1^*) \in \mathcal{P}_0 \times \mathcal{P}_1$ is called \textit{least favorable distribution (LFD) pair} in terms of risk $R$ for testing $\mathcal{P}_0$ vs.\ $\mathcal{P}_1$, if for every likelihood ratio test $\phi$ between $P_0^*$ and $P_1^*$, $R(P_j^*,\phi) \geq R(P_j,\phi), ~\text{ for all } P_j \in \mathcal{P}_j,~ j=0,1.$
\end{definition}
Intuitively, an LFD pair makes the testing problem hardest: any test calibrated for $(P_0^*,P_1^*)$ performs at least as well (has smaller errors of both types)  against all other distributions in the respective classes. LFD pairs are known to exist for various parametric and nonparametric models, e.g., the monotone likelihood ratio (MLR) families, and Huber's robust models (including $\epsilon$ neighborhood model, total variation model, and more generally with 2-alternating capacities \citep{huber1973minimax}). 

More recently, \cite{saha2025huber} has characterized the GROW e-variable for composite testing in the presence of an LFD pair.
 If  $(P_0^*, P_1^*)$ is an LFD pair in terms of the risk defined in \eqref{risk} for testing $\mathcal{P}_0$ vs.\ $\mathcal{P}_1$, they show that
   $\frac{dP_{1}^*(X)}{dP_{0}^*(X)}$ is a GROW e-variable for testing $\mathcal{P}_0$ vs.\ $\mathcal{P}_1$. That is, if $\mathcal{E}(\mathcal{P}_0)$ denotes the set of all e-variables for $\mathcal{P}_0$, we have $\frac{dP_{1}^*(X)}{dP_{0}^*(X)}\in\mathcal{E}(\mathcal{P}_0)$ and 
  \begin{align*}
     & \sup_{E\in \mathcal{E}(\mathcal{P}_0)} \inf_{P_1\in \mathcal{P}_1}\mathbb{E}_{P_1}\!\left(\log E\right)=\inf_{P_1\in \mathcal{P}_1}\mathbb{E}_{P_1}\!\left(\log\tfrac{d P_1^*(X)}{d P_0^*(X)}\right)\\
     &= \kl(P_1^*,P_0^*)=\inf_{(P_0,P_1)\in\mathcal{P}_0 \times \mathcal{P}_1}\kl(P_1,P_0).
  \end{align*}
The results above establish that, in unconstrained settings, log-optimal e-variables are likelihood ratios between least favorable distributions (when they exist) in composite problems. 
The central question of the present paper is: 
\begin{quote}
\emph{What happens if we impose additional structural constraints on the e-variable, such as privacy, quantization, boundedness, or convex integral constraints (e.g., moment bounds)?}
\end{quote}

\subsection{Our contribution}
For various realistic constraint sets, we show that the constrained log-optimal e-variable is obtained by some post-processing of the unconstrained log-optimal e-variable.

Specifically, we derive exact, closed-form solutions for the log-optimal e-variable under four distinct constraint classes, revealing a shared backbone:
\begin{enumerate}
    \item \textbf{Local Differential Privacy:} An optimality result of \cite{kairouz2016extremal} for the binary mechanism focuses on high-privacy regimes (small $\epsilon$) for general alphabets, but in the binary case, the optimality holds universally for all $\epsilon$. As a consequence, we can derive the optimal LDP Kelly bet, summarized in \Cref{alg:kelly-bet}. We introduce a slightly different optimization problem and characterize its solution in \Cref{thm:opt-binary-simple}, which is at least as good as the binary mechanism of \cite{kairouz2016extremal}, for any $\epsilon$, $P$, and $Q$. Furthermore, we generalize this binary mechanism and its corresponding e-variable to the composite case in \Cref{thm:bin-opt-comp}, demonstrating that the optimal private e-variable remains a randomized post-processing of the unconstrained optimal e-variable. 
    \item \textbf{Boundedness:} We prove in \Cref{thm:bounded} that enforcing an almost-sure boundedness constraint $E \in [c_1, c_2]$ results in a simple deterministic truncation (clipping) of a suitably normalized likelihood ratio for simple vs. simple hypothesis.
    \item \textbf{Quantization:} Under the constraint that the e-variable can take only two distinct values, we show in \Cref{thm:quantization} that the optimal solution is a step-function defined entirely by a threshold on the likelihood ratio. 
    \item \textbf{Bounded convex integral constraints:} We analyze the optimal e-variable under a broad class of convex integral constraints, where the expected value of $\phi(E)$, for a convex function $\phi$ is bounded under the null hypothesis. We establish in \Cref{thm:convex-constraint} that the optimal restricted e-variable is a strictly monotone transform of the unconstrained likelihood ratio.
\end{enumerate}
A key structural insight emerges from these constraint classes: in the simple-vs-simple setting, the constrained growth-optimal e-variable is always a non-decreasing transformation of the likelihood ratio. In \Cref{thm:comp-gen}, we extend this principle to composite testing. Specifically, whenever a least favorable distribution (LFD) pair exists, and the simple-versus-simple optimizer is a non-decreasing transformation of the likelihood ratio, the constrained growth-optimal e-variable for the composite problem is likewise a transformation of the unconstrained optimal e-variable. While we state all the optimality results with respect to standard logarithmic utility, we provide proof arguments for \Cref{thm:bounded,thm:convex-constraint,thm:comp-gen} under a broader class of utility functions in \Cref{thm:bounded-gen,thm:convex-constraint-gen,thm:gen}.


\paragraph{Outline.}
The rest of the paper is organized as follows. In \Cref{sec:ldp}, we derive log-optimal e-variable under LDP constraint with a binary mechanism, first for the simple vs simple setting and then extend it to the composite problem with LFDs.
In \Cref{sec:bounded,sec:quantized,sec:moment}, we find out the log-optimal e-variable for the simple vs. simple setting under quantization, boundedness, and convex integral constraints, respectively. In \Cref{sec:composite},
we generalize the quantization, boundedness,
and convex integral constrained results to composite hypothesis testing problems with LFDs. We conclude in \Cref{sec:conc}, following a discussion in \Cref{sec:discussion} which provides a discussion on possible extensions and limitations. We conclude in \Cref{sec:conc}, preceded by \Cref{sec:discussion}, which outlines possible extensions and limitations of our framework. All theorem proofs are provided in the appendix.

\section{Local Differential Privacy Constraint}
\label{sec:ldp}
 Local differential privacy \citep{dwork2006calibrating,kasiviswanathan2011can,duchi2023local} is defined through a conditional probability distribution $Q(\cdot|x)$, which represents the mechanism that randomizes an input $x \in \mathcal{X}$ to an output $y \in \mathcal{Y}$. We say that a mechanism $\Q$ is
{\em $\varepsilon$-locally differentially private} if
\begin{eqnarray}
	\sup_{S\subset \cY, x,x' \in\cX} \frac{\Q(S|\x)}{\Q(S|\x')} \;\leq\; e^\varepsilon \;,
	\label{eq:defLDP}
\end{eqnarray}
where $\Q(S|x) := \prob(Y_i\in S | X_i=x)$ represents the privatization mechanism.
This ensures that for small values of $\varepsilon$, given a privatized data $Y_i$, it is (almost) equally likely to have come from
any data, i.e. $x$ or $x'$. A small value of $\varepsilon$ means that we require a high level of privacy and a large value corresponds to a low level of privacy.
LDP has become a standard model in privacy-preserving data collection, with diverse applications in surveys, federated systems, etc.

 We first consider simple hypotheses $P_0$ vs.\ $P_1$ under the LDP constraint. Assume that $P_0$ and $P_1$ have densities $p_0$ and $p_1$ with respect to some common dominating measure.
 We shall finally generalize it to the composite testing problem, where a least favourable distribution pair exists.

\cite{kairouz2016extremal,pensia2024simple} address a similar problem, but they only consider simple vs simple testing with distributions having finite support. And the optimality results in \cite{kairouz2016extremal} focus on high or low privacy regimes (small or large $\epsilon$) only.
More broadly, there is a substantial literature on LDP hypothesis testing, including goodness-of-fit testing \cite{gaboardi2016differentially,lam2022minimax}, two-sample testing \cite{mun2024minimax}. These works adopt minimax or Neyman-Pearson risk criteria,
measuring performance through Type I/II errors, and they do not provide the safe, anytime-valid guarantees in composite hypothesis testing problems, as enabled by the e-value framework.
Recent work by \citet{csillag2025differentially} introduces differentially private e-values under the \emph{global} differential privacy model, where a trusted curator perturbs aggregate statistics to guarantee privacy for the entire dataset. 
In contrast, the present paper studies e-values under \emph{local} differential privacy (LDP), a strictly stronger privacy model in which each data point is randomized independently at the source. To the best of our knowledge, our work is the first to characterize \emph{growth-rate optimal LDP e-values} under the constraint that the privatized output is binary for composite testing problems. 

We are interested in releasing a differentially private version of $\X$ represented by $\Y$. The e-value is then a function of the randomized output $Y$, not the raw data $X$. The random variable $\Y$
should preserve the information content of $\X$ as much as possible while
meeting the local differential privacy constraints.
 The output of the privatization mechanism $\Y$ is distributed according to the induced marginal $M^Q_i$ given by
\begin{eqnarray*}
	M_{P_i}^Q(y) &=& \int_{\x\in\cX} Q(y|\x)  dP_i(\x) , \text{ for } y\in \cY.
\end{eqnarray*}
For sufficiently small $\epsilon$ and when the distributions $P_0$ and $P_1$ have finite support,
\cite{kairouz2016extremal} characterizes the optimal solution to
\begin{equation}
\label{eq:opti}
\begin{aligned}
& \underset{\Q}{\text{maximize}}~\kl(M_{P_1}^Q \| M_{P_0}^Q), ~ \text{subject to}
~ \Q \in \mathcal{\D}_{\varepsilon},
\end{aligned}
\end{equation}
where $\mathcal{\D}_{\varepsilon}$ is the set of all $\varepsilon$-locally
differentially private mechanisms satisfying \eqref{eq:defLDP}.   Given the mechanism $Q\in\mathcal{\D}_{\varepsilon}$, and with $Y\sim  Q(.|X)$, we know that the log-optimal e-variable is the likelihood ratio of the alternative and null distribution of $Y$:
$E = \frac{dM_{P_1}^Q(Y)}{dM_{P_0}^Q(Y)}.$
Therefore, the 
 \textit{log-optimal $\varepsilon$-LDP e-variable} for testing $P_0$ vs. $P_1$ is defined as the likelihood ratio
\begin{equation}
    E = \frac{dM_{P_1}^{Q^*}(Y)}{dM_{P_0}^{Q^*}(Y)},
\end{equation}
where $Y\sim \Q^*(.|X)$, and $Q^*$ is a solution to \eqref{eq:opti}.
For a given $\PP_0$ and $\PP_1$, the {\em binary mechanism}  \citep{kairouz2016extremal} is defined as a staircase mechanism with only two outputs
$\Y\in\{0,1\}$ satisfying 
\begin{eqnarray}
\Q(0|x) \,=\, \left\{
\begin{array}{rl}
	\dfrac{e^\varepsilon}{1+e^\varepsilon}& \text{ if } p_0(x)\geq p_1(x)\;,\\
	\dfrac{1}{1+e^\varepsilon}& \text{ if } p_0(x)< p_1(x)\;,\\
\end{array}
\right.\;\;\;
\label{eq:defhypbin}
\end{eqnarray}
and $\Q(1|x) \,=\, 1-\Q(0|x).$
And they have
established that it is the optimal mechanism when a high level of privacy is
required. The following theorem is a direct consequence of Theorem 5 of \cite{kairouz2016extremal}.
\begin{theorem}
\label{thm:hypbin}
	For any $\PP_0$ and $\PP_1$ with finite support,
	there exists a positive $\varepsilon^*$ that depends on  $\PP_0$ and $\PP_1$ such that
	the binary mechanism $Q$ defined in \eqref{eq:defhypbin} solves \eqref{eq:opti}, i.e., it maximizes the KL-divergence between the induced marginals over all
	$\varepsilon$-LDP mechanisms. Therefore, the log-optimal $\varepsilon$-LDP e-variable for $\PP_0$ vs.\ $\PP_1$, when $\varepsilon<\varepsilon^*$, is given by
    \begin{equation*}
        \label{opt-evalue}
        E=\frac{dM_{P_1}^Q(Y)}{dM_{P_0}^Q(Y)}=\begin{cases}
            & \frac{e^\epsilon P_1(T) + (1-P_1(T))}{ e^\epsilon ( 1-P_0(T))+ P_0(T)} ,\text{ if } Y=0,\\
            &  \frac{ e^\epsilon P_1(T^c) + (1-P_1(T^c))}{ e^\epsilon ( 1-P_0(T^c))+ P_0(T^c)} ,\text{ if } Y=1,
        \end{cases}
    \end{equation*}
    where $Y\sim \Q(.|x)$ and $T=\{x\in\mathcal{X}:p_0(\x)\geq p_1(\x)\}.$
\end{theorem}
While the above optimality result from \cite{kairouz2016extremal} for the binary mechanism focuses on high-privacy regimes (small $\epsilon$) for general finite alphabets, next we show that for the binary case, the optimality holds universally for all $\epsilon$. As a consequence, we derive the log-optimal LDP Kelly bet for testing a fair coin against a biased coin with $q>1/2$, summarized in \Cref{alg:kelly-bet}. Classical Kelly betting \citep{kelly1956new} chooses $(2q-1)$ fraction of wealth that maximizes the expected logarithmic growth rate of capital under the alternative. Under local differential privacy, the bettor does not observe the true outcome but only a randomized response version of it. The optimal strategy, therefore, adjusts the betting fraction to account precisely for the privacy-induced loss of information.
At each round, the raw Bernoulli observation is privatized via randomized response, and the skeptic updates wealth using a privacy-adjusted betting fraction $f^* =(2q-1) \cdot \frac{e^\epsilon - 1}{e^\epsilon + 1}$.
Relative to the unconstrained Kelly strategy , the optimal betting fraction under $\varepsilon$-LDP is attenuated by the multiplicative factor $\frac{e^\epsilon - 1}{e^\epsilon + 1}$.
\begin{theorem}
\label{thm:kelly-bet}
Let  $P_0$ and $P_1$ are Ber$(p_0)$ and Ber$(p_1)$ respectively. Then, for any privacy budget $\epsilon > 0$, the binary mechanism $Q$ defined as:
\begin{equation}
\label{eq:mech-bernoullli}
    Q(y|x) = \begin{cases} 
    \frac{e^\epsilon}{1+e^\epsilon} & \text{if } y=x \\
    \frac{1}{1+e^\epsilon} & \text{if } y \ne x
    \end{cases}
\end{equation}
 solves \eqref{eq:opti}. Therefore, the e-variable in \eqref{eq:dp-kelly-bet} is the log-optimal $\varepsilon$-LDP e-variable for testing Ber$(1/2)$ against Ber$(q)$.
\end{theorem}
\begin{algorithm}[!h]
\caption{The LDP Kelly Bet}
\begin{algorithmic}[1]
\State \textbf{Input:} $q > 1/2$, Privacy budget $\epsilon$.
\State \textbf{Initialize:} Skeptic's initial capital $W_0 \gets 1$.
\State \textbf{Calculate bet fraction:} 
\[
f^* \gets (2q-1) \cdot \frac{e^\epsilon - 1}{e^\epsilon + 1}
\]
\Comment{Optimal Kelly bet adjusted for privacy constraint}

\For{$t = 1, 2, \dots$}
    \State \textbf{Reality:} Nature draws $X_t \in \{0, 1\}$.
    \State \quad Under $H_0$: $X_t \sim \text{Ber}(1/2)$.
    \State \quad Under $H_1$: $X_t \sim \text{Ber}(q)$.
    
    \State \textbf{Mechanism:} Release privatized bit $Y_t$.
    \[
    Y_t = 
    \begin{cases} 
    X_t & \text{w.p. } \frac{e^\epsilon}{1+e^\epsilon} \\
    1 - X_t & \text{w.p. } \frac{1}{1+e^\epsilon}
    \end{cases}
    \]
    
    \State \textbf{Update Wealth:} Skeptic updates capital based on outcome $Y_t$.
    \begin{equation}
    \label{eq:dp-kelly-bet}
        e(Y_t) = 1 + f^* \cdot (2Y_t - 1),
    \end{equation}
    \[
    W_t \gets W_{t-1} \times e(Y_t).
    \]
\EndFor
\end{algorithmic}
\label{alg:kelly-bet}
\end{algorithm}
The binary mechanism is a popular choice for many applications, yet it is not optimal for general $\epsilon$, $P$, and $Q$ when the inputs are not Bernoulli. It is also unknown how small $\epsilon$ should be to achieve optimality. Therefore, in this paper, we focus on characterizing the optimal solution to
\begin{equation}
\label{eq:opti-binary}
\begin{aligned}
& \underset{\Q}{\text{maximize}}~ \kl(M_{P_1}^Q \| M_{P_0}^Q),~\text{subject to}~ \Q \in \mathcal{\D}_{\varepsilon}^{(2)},
\end{aligned}
\end{equation}
where $\mathcal{\D}_{\varepsilon}^{(2)}=\mathcal{\D}_{\varepsilon}\cap\{Q:\text{support}(Q)=\{0,1\}\}$ is the set of all $\varepsilon$-locally
differentially private mechanisms satisfying \eqref{eq:defLDP} that output only binary values.
 The \textit{log-optimal $\varepsilon$-LDP e-variable}  for testing $P_0$ vs. $P_1$, under the additional constraint that $\text{support}(Q)=\{0,1\}$ is defined as
\begin{equation}
    E = \frac{dM_{P_1}^{Q^*}(Y)}{dM_{P_0}^{Q^*}(Y)},
\end{equation}
where $Y\sim \Q^*(.|X)$, and $Q^*$ is a solution to \eqref{eq:opti-binary}.
Note that, by definition, the solution to \eqref{eq:opti-binary} is at least as good as the binary mechanism of \cite{kairouz2016extremal} defined in \eqref{eq:defhypbin}, for any $\epsilon$, $P$, and $Q$.

\subsection{Optimal $\epsilon$-LDP Binary Mechanism}
The objective is to maximize the KL divergence between the induced output distributions under $H_1$ and $H_0$. Let $m_{P_1}^Q$ and $m_{P_0}^Q$ denote the probability of outputting $1$ under $H_1$ and $H_0$, respectively:
\[ m_{P_1}^Q = \int_{x\in\mathcal{X}} Q(1| x) dP_1(x),~ m_{P_0}^Q = \int_{x\in\mathcal{X}} Q(1| x) dP_0(x).\]
Since we are restricted to $\Q \in \mathcal{\D}_{\varepsilon}^{(2)}$, the marginal distributions of $Y$ under the null and alternative are $M_{P_0}^Q=\text{Ber}( m_{P_0}^Q)$ and $M_{P_1}^Q=\text{Ber}( m_{P_1}^Q)$ respectively.
Under the constraint \eqref{eq:defLDP}, we want to maximize $\kl(M_{P_1}^Q \| M_{P_0}^Q)$, which can be written as a function of $Q$:
\[ J(Q) = m_{P_1}^Q\log \frac{m_{P_1}^Q}{m_{P_0}^Q} + (1 - m_{P_1}^Q) \log \frac{1 - m_{P_1}^Q}{1 - m_{P_0}^Q}.\]
 In other words, for arbitrary $P_0,P_1$ and $\epsilon\in(0,\infty),$ we are interested in characterizing the optimal solution to \eqref{eq:opti-binary}.
\begin{theorem}
\label{thm:opt-binary-simple}
For any distributions $P_0, P_1$ and any privacy budget $\epsilon > 0$, the binary mechanism 
\begin{equation}
\label{eq:binary-opt-simple}
    Q(1\mid x) = \begin{cases} 
\frac{e^\epsilon}{e^\epsilon + 1} & \text{if } \frac{dP_1(x)}{dP_0(x)} > t \\
\frac{1}{e^\epsilon + 1} & \text{if } \frac{dP_1(x)}{dP_0(x)} \le t \end{cases}
\end{equation}
solves \eqref{eq:opti-binary}, i.e., it maximizes the KL-divergence between the induced marginals over all
	$\varepsilon$-LDP \textit{binary} mechanisms, where the threshold $t$ is the solution to 
\[
t=\frac{m_{P_1}^{Q} - m_{P_0}^{Q}}{m_{P_0}^{Q}(1-m_{P_0}^{Q})\log\left( \frac{m_{P_1}^{Q}(1-m_{P_0}^{Q})}{m_{P_0}^{Q}(1-m_{P_1}^{Q})}\right)}.
\]
Therefore, the log-optimal e-variable under the same constraints is
  \begin{equation}
        \label{opt-evalue-simple}
        E=\frac{dM_{P_1}^Q(Y)}{dM_{P_0}^Q(Y)}=\begin{cases}
            & v_0 ,~~\text{ if } Y=0,\\
            & v_1 ,~~\text{ if } Y=1,
        \end{cases}
    \end{equation}
    where $Y\sim \Q(.|x)$, $v_0=\frac{ e^\epsilon P_1(T_{\text{opt}}) + (1-P_1(T_{\text{opt}}))}{ e^\epsilon ( 1-P_0(T_{\text{opt}}))+ P_0(T_{\text{opt}})}$, $v_1=\frac{ e^\epsilon P_1(T^c_{\text{opt}}) + (1-P_1(T^c_{\text{opt}}))}{ e^\epsilon ( 1-P_0(T^c_{\text{opt}}))+ P_0(T^c_{\text{opt}})}$,  and $T_{\text{opt}}=\{x\in\mathcal{X}:\frac{dP_1(x)}{dP_0(x)} \le t\},$ with $\Q(.|x)$ and $t$ as defined above.
\end{theorem}
We remark that the proof of the above result is nontrivial and fundamentally different from the techniques developed in \cite{kairouz2016extremal}. Their analysis relies heavily on combinatorial and extremal arguments tailored to finite output alphabets and unconstrained optimization. In contrast, our setting introduces an additional structural constraint and applies to arbitrary distributions $P_0, P_1$ and $\epsilon > 0$.

Solving for the optimal threshold $t$ exactly is generally not possible in closed form for arbitrary distributions $P_0$ and $P_1$. However, one can solve it easily using numerical iteration (fixed-point iteration). Note that the log-optimal $\varepsilon-$LDP e-variable $E$ under the additional constraint that $\text{support}(Q)=\{0,1\}$ (i.e., $Y$ can take values 0 and 1 only), defined in \eqref{opt-evalue-simple}, can be written as a function of the unconstrained log-optimal e-variable, $L(X)=\frac{dP_1(X)}{dP_0(X)}$ and an independent uniform random variable $U$:
\begin{align*}
    E&= v_0\Biggl(\mathds{1}\left(L(X)> t, U <\dfrac{e^\varepsilon}{1+e^\varepsilon}\right)\\
    &+\mathds{1}\left(L(X)\leq t, U <\dfrac{1}{1+e^\varepsilon}\right)\Biggr)  \\
    &+v_1\Biggl(\mathds{1}\left(L(X)> t, U\geq\dfrac{1}{1+e^\varepsilon}\right)\\
    &+\mathds{1}\left(L(X)\leq t, U \geq\dfrac{e^\varepsilon}{1+e^\varepsilon}\right)\Biggr).
\end{align*}
\subsection{Extension to composite nulls and alternatives}
We now generalize the result to the composite testing problem. Suppose $(P_0^*, P_1^*)$ is a least favorable distribution (LFD) pair in terms of the risk defined in \eqref{risk} for testing $\mathcal{P}_0$ vs.\ $\mathcal{P}_1$. 
We want to characterize the optimal solution to
\begin{equation}
\label{eq:opti-binary-lfd}
\begin{aligned}
& \underset{\Q \in \mathcal{\D}_{\varepsilon}^{(2)}}{\text{maximize}} 
& & \underset{{(P_0, P_1)\in\mathcal{P}_0\times\mathcal{P}_1}}{\text{minimize}}\kl(M_{P_1}^Q \| M_{P_0}^Q),
\end{aligned}
\end{equation}
where $\mathcal{\D}_{\varepsilon}^{(2)}=\mathcal{\D}_{\varepsilon}\cap\{Q:\text{support}(Q)=\{0,1\}\}$ is the set of all $\varepsilon$-locally
differentially private mechanisms satisfying \eqref{eq:defLDP} that output only binary values.

Now consider the binary mechanism in \eqref{eq:binary-opt-simple} with the LFD pair as follows.
\begin{eqnarray}
\Q^*(0|x) \,=\, \left\{
\begin{array}{rl}
	\dfrac{e^\varepsilon}{1+e^\varepsilon}& \text{ if } \frac{dP_1^*(x)}{dP_0^*(x)} > t^*\;,\\
	\dfrac{1}{1+e^\varepsilon}& \text{ if } \frac{dP_1^*(x)}{dP_0^*(x)} \leq t^*\;,\\
\end{array}
\right.\;\;\;
\label{eq:eq:binary-opt-lfd}
\end{eqnarray}
and $\Q^*(1|x)=1-\Q^*(0|x),$
where the threshold $t^*$ is the solution to 
\begin{equation}
\label{eq:opt-threshold-comp}
    t^*=\frac{m_{P_1}^{Q^*} - m_{P_0}^{Q^*}}{m_{P_0}^{Q^*}(1-m_{P_0}^{Q^*})\log\left( \frac{m_{P_1}^{Q^*}(1-m_{P_0}^{Q^*})}{m_{P_0}^{Q^*}(1-m_{P_1}^{Q^*})}\right)}.
\end{equation}
Analogously define the values $v^*_0=\frac{ e^\epsilon P_1^*(T^*) + (1-P_1^*(T^*))}{ e^\epsilon ( 1-P_0^*(T^*))+ P_0^*(T^*)}$, $v_1=\frac{ e^\epsilon  (1-P_1^*(T^*))+P_1^*(T^*)}{ e^\epsilon  P_0^*(T^*)+( 1-P_0^*(T^*))}$ and
\begin{equation}
        \label{opt-evalue-comp}
        E^*=\frac{dM_{P_1^*}^{Q^*}(Y)}{dM_{P_0^*}^{Q^*}(Y)}=\begin{cases}
            & v^*_0 ,\text{ if } Y=0,\\
            & v^*_1 ,\text{ if } Y=1,
        \end{cases}
    \end{equation}
     where $Y\sim \Q^*(.|x)$ satisfies $\epsilon$-LDP constraint, and $T^*=\{x\in\mathcal{X}:\frac{dP_1^*(x)}{dP_0^*(x)} \leq t^*\}$.
\begin{theorem}
\label{thm:bin-opt-comp}
Suppose $(P_0^*, P_1^*)$ is a least favorable distribution (LFD) pair in terms of the risk defined in \eqref{risk} for testing $\mathcal{P}_0$ vs.\ $\mathcal{P}_1$. Then, $\Q^* \in \mathcal{\D}_{\varepsilon}^{(2)}$ and $(P_0^*, P_1^*)\in\mathcal{P}_0\times\mathcal{P}_1$ solves \eqref{eq:opti-binary-lfd}, i.e.,
\begin{align*}
     & \sup_{Q \in \mathcal{D}_{\epsilon}^{(2)}} \inf_{\substack{P_1 \in \mathcal{P}_1 \\ P_0 \in \mathcal{P}_0}} D_{KL}(M_{P_1}^{Q} \| M_{P_0}^{Q}) = D_{KL}(M_{P_1^*}^{Q^*} \| M_{P_0^*}^{Q^*})\\
      &= \sup_{Q \in \mathcal{D}_{\epsilon}^{(2)}} D_{KL}(M_{P_1^*}^{Q} \| M_{P_0^*}^{Q}).
\end{align*}
Moreover, $E^*$ defined in \eqref{opt-evalue-comp}
   is an e-variable under the induced composite null $\mathcal{P}_0^{Q^*}=\{M_{P_0}^{Q^*}:P_0\in \mathcal{P}_0\}$ and it is log-optimal against the induced composite alternative $\mathcal{P}_1^{Q^*}=\{M_{P_1}^{Q^*}:P_1\in \mathcal{P}_1\}$:
  \begin{align*}
     & \sup_{E\in \mathcal{E}(\mathcal{P}_0^{Q^*})} \inf_{M\in\mathcal{P}_1^{Q^*}}\mathbb{E}_{M}\!\left(\log E\right)=\inf_{M\in\mathcal{P}_1^{Q^*}}\mathbb{E}_{M}\!\left(\log E^*\right)\\
     &=\kl(M_{P_1^*}^{\Q^*},M_{P_0^*}^{\Q^*}),
  \end{align*}
  where $ \mathcal{E}(\mathcal{P}_0^{Q^*})$ denotes the set of all e-variables for $\mathcal{P}_0^{Q^*}$  and $Q^*$ is as defined in \eqref{eq:eq:binary-opt-lfd}.
\end{theorem}
Note that the log-optimal e-variable $E^*$ defined in \eqref{opt-evalue-comp} can be written as a function of the unconstrained log-optimal e-variable, $L^*(X)=\frac{dP_1^*(X)}{dP_0^*(X)}$ and an independent uniform random variable $U$:
\begin{align*}
    E^*&= v_0^*\Bigg(\mathds{1}\left(L^*(X)> t^*, U <\dfrac{e^\varepsilon}{1+e^\varepsilon}\right)\\
    &+\mathds{1}\left(L^*(X)\leq t^*, U <\dfrac{1}{1+e^\varepsilon}\right)\Bigg)  \\
    &+v_1^*\Bigg(\mathds{1}\left(L^*(X)> t^*, U \geq\dfrac{1}{1+e^\varepsilon}\right)\\
    &+\mathds{1}\left(L^*(X)\leq t^*, U \geq\dfrac{e^\varepsilon}{1+e^\varepsilon}\right)\Bigg).
\end{align*}
Thus, the constrained optimal solution in this case is a random transformation of the unconstrained optimal solution $L^*(X)$. In the next few sections, we validate this ``optimize-then-constraint'' principle through other types of constraints.

\section{Quantization constraint}
\label{sec:quantized}
Communication, storage, or hardware constraints may require the statistic to lie in a finite set.
Under communication/quantization constraints, a recent line of work, originating in \cite{tsitsiklis1989decentralized}, established minimax optimal rates for a variety of problems,
including distribution estimation and identity testing \cite{pmlr-v75-han18a,chen2021pointwise}, simple hypothesis testing \cite{pensia2023communication}. Among these, our setup is closest to \citet{pensia2023communication},
which analyzes sample complexity for simple hypothesis testing, when each sample is quantized before transmission, provides bounds on sample complexity, However, their setup is fundamentally different in that they focus on fixed sample settings, controlling type-I and type-II errors, and do not provide the safe, anytime-valid guarantees as enabled by the e-value framework. Moreover, their work is limited to simple hypothesis testing, while we generalize our e-value-based framework to composite testing in the presence of LFDs in \Cref{sec:composite}.

In this section, we study the simple vs simple hypothesis testing problem with the simplest nontrivial case:
binary (two-level) quantization. However, it can be generalized to any finite set-size.
Formally, we seek the optimal random variable $E^*$ within the class of \textbf{quantized e-variables} $\mathcal{E}_{2}$. A random variable $E \in \mathcal{E}_2$ if:
\begin{enumerate}
    \item \textbf{Quantization Constraint:} $E$ takes at most two distinct values $\{u_0, u_1\}$.
    \item \textbf{e-variable Constraint:} $\mathbb{E}_{P_0}[E] \le 1$.
\end{enumerate}
The optimization problem is:
\[
\sup_{E \in \mathcal{E}_2} \mathbb{E}_{P_1}[\log E].
\]
The following theorem characterizes the solution.
\begin{theorem}
\label{thm:quantization}
The solution to the above optimization problem is given by:
\[
E^* = \begin{cases} 
u_1 & \text{if } L(X) > t^*, \\
u_0 & \text{if } L(X)\le t^*, \end{cases}
\]
where $L(X)=\frac{dP_1(X)}{dP_0(X)}$ and the values are:
\[
u_1 = \frac{P_1(L(X) > t^*)}{P_0(L(X) > t^*)}, \quad u_0 = \frac{P_1(L(X) \le t^*)}{P_0(L(X) \le t^*)},
\]
and the threshold $t^*$ is $t^* = \frac{u_1 - u_0}{\log u_1 - \log u_0}.$
\end{theorem}

We reiterate that the constrained optimizer is obtained by applying a monotone transformation on the unconstrained optimal solution, i.e., the likelihood ratio.

\section{Boundedness constraint}
\label{sec:bounded}
Imposing almost-sure boundedness constraints on e-variables is natural in several settings.
From a robustness perspective, bounding a test statistic prevents excessive influence
from rare but extreme observations, a principle closely related to classical robust statistics
and influence-function control \cite{huber1964robust,hampel1974influence}.
Boundedness is equally fundamental in privacy-preserving analysis. In global differential privacy, sensitivity is defined through worst-case bounded changes in the statistic under single-point perturbations, and many standard mechanisms (e.g., Laplace or Gaussian mechanisms) require a bounded test statistic to ensure finite noise calibration
\cite{dwork2006calibrating,dwork2014algorithmic}. While recent work by \citet{csillag2025differentially} relies on e-values with bounded sensitivity for constructing globally differentially private tests, it does not provide a systematic way to construct the log-optimal bounded e-value. 

In this section, we consider testing $P_0$ against $P_1$ in the simple versus simple setting.
We define the admissible domain $\mathcal{D}$ as the set of all $P_0$-measurable random variables bounded almost surely within the interval $[c_1, c_2]$, where $0 \le c_1 \le 1 \le c_2 < \infty$:
\[
\mathcal{D} = \{ E : X \to \mathbb{R} \mid c_1 \le E \le c_2 \quad P_0\text{-a.s.} \}.
\]
We seek a random variable $E^*$ that solves the following constrained optimization problem:
\begin{align*}
\sup_{E \in \mathcal{D}} \quad  \mathbb{E}_{P_1}[\log E] \quad
\text{subject to} \quad & \mathbb{E}_{P_0}[E] \le 1.
\end{align*}
\begin{theorem}
\label{thm:bounded}
The solution $E^*$ to the optimization problem exists and is unique $P_0$-almost surely and is given by:
\[
E^* = \min\left( c_2, \max\left( c_1, \frac{L(X)}{\lambda^*} \right) \right)\quad P_0\text{-a.s.},
\]
where $L(X)=\frac{dP_1(X)}{dP_0(X)}$ is the likelihood ratio and $\lambda^* > 0$ is the constant such that $\mathbb{E}_{P_0}[E^*] = 1$.
\end{theorem}
Thus, the constrained optimizer is obtained by
\emph{post-processing} the unconstrained optimizer, i.e., the likelihood ratio.
This exemplifies the optimize–then–constrain principle:
first compute the log-optimal e-variable,
then impose the structural constraint via a transformation.

\section{Bounded convex integral constraints}
\label{sec:moment}
If the likelihood ratio possesses heavy tails, the resulting unconstrained e-variable may exhibit massive or even infinite variance, leading to highly unstable wealth trajectories. 
To enforce statistical stability, it is often crucial to restrict the e-variable space to those satisfying bounded moment constraints (such as a bounded second moment, $\mathbb{E}_{P_0}[E^2] \le C$). To make it more general, we consider a broader class of \emph{convex integral constraints}. 

Let $\phi : (0, \infty) \to \mathbb{R}$ represent a convex penalty function that penalizes extreme values. We seek to find the growth-optimal e-variable subject to the standard validity constraint and an additional bound on its expected penalty under the null hypothesis.
We consider the optimization problem
\begin{equation}
\begin{aligned}
& \underset{E}{\text{maximize}}
& & \mathbb{E}_{P_1}[\log E] \\
& \text{subject to}
& & \mathbb{E}_{P_0}[E] \le 1, \quad \mathbb{E}_{P_0}[\phi(E)] \le C.
\end{aligned}
\end{equation}
for some fixed $C \in \mathbb{R}$ such that the feasible set is non-empty. The following theorem demonstrates that as long as $\phi$ is strictly convex, continuously differentiable, and superlinear, the optimal e-variable always takes the form of a monotonic transform of the likelihood ratio $L=\frac{dP_1}{dP_0}$.
\begin{theorem}
\label{thm:convex-constraint}
 Let $\phi : (0, \infty) \to \mathbb{R}$ be strictly convex and superlinear: $\lim_{x \to \infty} \frac{\phi(x)}{x} = \infty.$
 Then for the above optimization problem,
\begin{itemize}
    \item[(i)] There exists a unique maximizer $E^*$ (up to $P_0$-a.s.).
    \item[(ii)] 
    There exists an increasing function $\psi : (0, \infty) \to (0, \infty)$ such that $E^* = \psi(L) \quad P_0\text{-a.s.}$
\end{itemize}
Moreover, if $\phi$ is differentiable, then there exist $\lambda \in \mathbb{R}$ and $\gamma \ge 0$ such that $\frac{L}{E^*} = \lambda + \gamma \phi'(E^*) ~~ P_0\text{-a.s.}$ 
\end{theorem}
The proof relies on several elegant, nontrivial arguments about closedness/compactness of the feasible set of e-values in the appropriate topology.

As a concrete illustration, consider the bounded second-moment constraint $\mathbb{E}_{P_0}[E^2] \le C$, by taking $\phi(x)=x^2$. In this case, it follows directly from the above theorem that
\begin{equation}
        E^* = \frac{\sqrt{\lambda^2 + 8\gamma L} - \lambda}{4\gamma},
    \end{equation}
where $\lambda \in \mathbb{R}$ and $\gamma \ge 0$ are chosen such that the constraints are satisfied with equality.

An analogous analysis applies to constraints imposed under the alternative (i.e., $\mathbb{E}_{P_1}[\phi(E)] \le C$) by observing the fact that $\mathbb{E}_{P_1}[\phi(E)]$ can be rewritten as $\mathbb{E}_{P_0}[L\phi(E)]$.

Notably, in this section as well as in the previous two sections, the optimal solution in the simple-versus-simple setting is always a monotone function of the likelihood ratio.

\section{Generalization to composite testing problem with LFD pair}
\label{sec:composite}
The results in the last three sections can be extended to composite hypothesis testing in the presence of a least favorable distribution (LFD) pair. They all share the same structure, where the optimal solution for the simple versus simple problem is a non-decreasing function of the likelihood ratio. In this section, we show that this structural property allows one to 
lift the simple-hypothesis solution to certain composite testing problems whenever a least favorable distribution (LFD) pair exists, which is a strict generalization of Theorem 2.1 of \cite{saha2025huber}.

Suppose we wish to test a composite null hypothesis $\mathcal{P}_0$ against a composite alternative $\mathcal{P}_1$, and $(P_0^*,P_1^*)$ is an LFD pair testing $\mathcal{P}_0$ against $\mathcal{P}_1$. Our objective is to find the log-optimal or growth rate optimal in worst-case (GROW) e-variable $E^*$ subject to a structural constraint (e.g., boundedness, binary outputs, or moment bounds). Let $\mathcal{E}^\prime(H_0) \subseteq \mathcal{E}(H_0)$ denote this constrained class of valid e-variables, for any null $H_0$. The optimization problem is given by:
\[
\sup_{E \in \mathcal{E}^\prime(\mathcal{P}_0)} \inf_{P \in \mathcal{P}_1} \mathbb{E}_{P}[\log E].
\]
 Let $L^* = \frac{dP_1^*}{dP_0^*}$ denote the likelihood ratio of the LFD pair.
\begin{theorem}
\label{thm:comp-gen}
Let $E^*$ be the optimal solution for the simple hypothesis pair $(P_0^*, P_1^*)$, defined as:
\[
E^* = \argmax_{E \in \mathcal{E}^\prime(\{P_0^*\})}  \mathbb{E}_{P_1^*}[\log E].
\]
Assume that it is of the form $E^*=\psi(L^*),$ for some non-decreasing function $\psi$.
Then $E^*\in\mathcal{E}^\prime(\mathcal{P}_0)$ and
\begin{align*}
    \sup_{E \in \mathcal{E}^\prime(\mathcal{P}_0)} \inf_{P_1 \in \mathcal{P}_1} \mathbb{E}_{P_1}[\log E] = &\inf_{P_1 \in \mathcal{P}_1} \mathbb{E}_{P_1}[\log E^*]\\
    &= \mathbb{E}_{P_1^*}[\log E^*].
\end{align*}
\end{theorem}
The theorem formalizes the following fact: if the constrained simple-vs-simple optimizer
is monotone in the likelihood ratio,
then it automatically extends to the composite problem
via the LFD pair.

\section{Discussion}
\label{sec:discussion}
\paragraph{Beyond logarithmic utility.}
While our results focus on log-optimality due to its wide recognition and close connection with sequential testing, betting, and information theory \cite{kelly1956new,cover1987log,breiman1961optimal,shafer2021testing,grunwald2020safe}, the principle is not strictly limited to the logarithmic utility. Motivated by recent works \cite{koning2024continuous,larsson2024numeraire} that considered broader utility functions, we note that the constrained optimization results established in \Cref{sec:bounded,sec:moment,sec:composite} seamlessly generalize to any utility function $U: (0, \infty) \to \mathbb{R}$ that is strictly increasing and strictly concave. These proofs do not rely on the specific logarithmic form; rather, they depend only on the monotonicity and concavity of the utility function. For completeness, we provide arguments under general utility functions in \Cref{proof-sec4,proof-sec5,proof-sec6}.

\paragraph{Counter-example with no LFD.}
\Cref{thm:comp-gen} strictly relies on the existence of an LFD pair. It is natural to ask whether the optimize-then-constrain principle survives if we relax this requirement. 
Let us consider Example 5.2 from \cite{larsson2024numeraire}, where the observation is $X \in [0, 1]$, the alternative ${Q}$ is Uniform(0,1), and the composite null is the bounded mean class $\mathcal{P} = \{ {P} : \mathbb{E}_{{P}}[X] \le \mu \}$ for some $\mu \in (0, 1/2)$. In this example, LFD does not exist, but the unconstrained log-optimal e-variable (a.k.a.  numeraire e-variable) does exist.
Suppose we impose an almost-sure boundedness constraint $E \le c$. If the optimize-then-constrain principle held, our \Cref{thm:bounded} with $c_1=0,c_2=c$ would suggest that the constrained optimal e-variable is a scaled truncation of the unconstrained numeraire $\left(1 + \lambda^*(X- \mu)\right)$:
\begin{equation}
\label{q:e-star}
    E^* = \min\left(c, \frac{1}{\gamma} \left(1 + \lambda^*(X- \mu)\right) \right),
\end{equation}
where $\lambda^*\in(0,1/\mu)$ and $\lambda^*=\displaystyle\argmax_{\lambda \ge 0} \mathbb{E}_{\mathbb{Q}}[\log(1+\lambda(X-\mu))]$.
Now, consider the following candidate e-variable
\begin{equation}
\label{eq:e-prime}
    E^\prime = \min\left(c, 1 + \lambda_{new}(X- \mu) \right),
\end{equation}
where $\lambda_{new}\in\displaystyle\argmax_{\lambda \ge 0} \mathbb{E}_{\mathbb{Q}}[\log(\min\{c,1 + \lambda(X-\mu)\})]$.
Then, \Cref{prop:counterexample} shows that one can choose $\mu$ and $c$ such that $E^\prime$ has a strictly larger growth rate than that of $E^*$.


When the alternative is also composite, even the existence of the unconstrained growth-optimal e-variable is not guaranteed for the composite null we considered in this example. In general, in the absence of an LFD pair, the unconstrained log-optimal e-value is not guaranteed to exist, and therefore, existence under constraints is not guaranteed either. Our counterexample suggests that in such cases, the optimize–then–constrain principle need not hold, and it is unlikely that a single unifying structural principle can be established without further assumptions.
Consequently, in the absence of LFDs, each constraint class may have to be analyzed separately.

\section{Conclusion}
\label{sec:conc}
This paper studies growth-rate optimal e-variables under several structural constraints in composite hypothesis testing.
In the unconstrained setting, the optimizer is given by the likelihood ratio between least favorable distributions (LFDs). Our main contribution is to determine the constrained optimal e-variables for several natural constraint classes, demonstrating that the constrained optimal solution is a transformation of the unconstrained optimal solution.
Our analysis crucially relies on the existence of least favorable distributions.
When no LFD exists, the log-optimal e-variable for both composite null and alternative is still an open question even for the unconstrained case. 
We illustrated through a counterexample (with a composite null and simple alternative), showing that the ``optimize–then–constrain’' principle can fail in the absence of an LFD pair. A precise characterization of the constraint classes and composite hypothesis families under which such a structural principle continues to hold is an open future direction.

\bibliographystyle{apalike}
\bibliography{ref}

\appendix

\section{Omitted proofs}

\subsection{Omitted proofs from Section 2}
\begin{proof}[Proof of \Cref{thm:kelly-bet}]

First, we must establish that maximizing the KL-divergence fits the optimization framework of Theorem 2 of \cite{kairouz2016extremal}.
They define the utility maximization problem as maximizing $U(Q) = \sum_{y \in \mathcal{Y}} \mu(Q_y)$, where $\mu$ is a sublinear function (convex and homogeneous). KL as Sublinear: $$D_{f}(M_0 \| M_1) = \sum_{y} \mu(Q_y),$$
where $\mu(Q_y) = (P_1^T Q_y) \log\left( \frac{P_0^T Q_y}{P_1^T Q_y} \right)$.
 This function $\mu$ is convex with respect to the mechanism column $Q_y$ because the function $\phi(z,t) = t f(z/t)$ is jointly convex.
So, the KL-divergence objective satisfies the conditions (sublinearity) required for Theorem 2 of \cite{kairouz2016extremal}, which states that for any sublinear function $\mu$ and any $\epsilon \ge 0$, there exists an optimal mechanism $Q^*$ that satisfies: (a) The output alphabet size is at most the input alphabet size: $|\mathcal{Y}| \le |\mathcal{X}|$, and (b) For all $y, x, x'$, the log-likelihood ratio $\left| \ln \frac{Q^*(y|x)}{Q^*(y|x')} \right|\in \{0, \epsilon\}$.
 Now we apply this general result to the specific Bernoulli case, where the input alphabet $\mathcal{X} = \{0, 1\}$.
 
 Step A: From Theorem 2(a), the optimal mechanism has output size $|\mathcal{Y}| \le |\mathcal{X}| = 2$. This means we only need to search for mechanisms with binary outputs ($Y \in \{0, 1\}$).
 
 Step B: Theorem 2(b) restricts the optimal mechanism to satisfy  $\frac{Q^*(y|x)}{Q^*(y|x')} \in \{1, e^\epsilon\}$. Also, note that, to maximize utility, we have $\frac{Q^*(y|x)}{Q^*(y|x')}=e^\epsilon,$ for $x\neq x'$.
 Let $Q(0|1) = \alpha$. Then $Q(0|0) = \alpha e^\epsilon$. Let $Q(1|0) = \beta$. Then $Q(1|1) = \beta e^\epsilon$. Applying probability constraints
 $$\alpha e^\epsilon + \beta = 1 \quad,\quad\alpha + \beta e^\epsilon = 1,$$
 and solving these equations, we get
  $\alpha=\beta = \frac{e^\epsilon}{1+e^\epsilon}$. This is exactly the definition of the $M^*$ provided in \eqref{eq:mech-bernoullli}.
\end{proof}

\begin{proof}[Proof of \Cref{thm:opt-binary-simple}]
    
\textbf{Step 1:}
The $\epsilon$-LDP condition imposes point-wise constraints on the range of the function $Q(1|x)$. Let $y_{max} = \sup_{x\in\mathcal{X}} Q(1|x)$ and $y_{min} = \inf_{x\in\mathcal{X}} Q(1|x)$. 

The constraint for the output $y=1$ implies:
\begin{equation*}
   \sup_{x,x'} \frac{Q(1|x)}{Q(1|x')} \le e^{\epsilon} \implies y_{max} \le y_{min}e^{\epsilon}.
\end{equation*}
The constraint for the output $y=0$ implies:
\begin{equation*}
   \sup_{x,x'} \frac{1-Q(1|x)}{1-Q(1|x')} \le e^{\epsilon} \implies 1-y_{min} \le (1-y_{max})e^{\epsilon}.
\end{equation*}
Thus, any admissible mechanism must satisfy $Q(1|x) \in [y_{min}, y_{max}]$ for all $x \in \mathcal{X}$. Let $\mathcal{Q}$ denote the convex set of all such measurable functions.

\vspace{1em}
\textbf{Step 2:}
For any mechanism $q \in \mathcal{Q}$, the induced marginal probabilities under the null and alternative measures are given by the linear functionals $m_{P_i}^q = \int_{\mathcal{X}} q(x) dP_i(x)$ for $i \in \{0, 1\}$. 

The objective to maximize is the KL-divergence between the induced Bernoulli distributions:
\begin{equation*}
    J(q) = m_{P_1}^q \log\left(\frac{m_{P_1}^q}{m_{P_0}^q}\right) + (1-m_{P_1}^q) \log\left(\frac{1-m_{P_1}^q}{1-m_{P_0}^q}\right).
\end{equation*}
We first establish that $J(q)$ is a convex functional of $q$. The KL-divergence between two probability distributions is fundamentally a jointly convex function of its arguments. Since the marginals $m_{P_1}^q$ and $m_{P_0}^q$ are strictly linear functionals of the mechanism $q$, the objective $J(q)$ is formed by the composition of a convex function with a linear map. Such a composition is always convex. Consequently, the objective $J(q)$ is a convex functional over the convex, bounded feasible set $\mathcal{Q}$. A fundamental property of convex optimization over a bounded set is that the maximum must occur at some extreme point of the feasible set.

Let $q^*$ be the globally optimal mechanism. Because $J(q)$ is strictly convex, the first-order necessary condition for $q^*$ to maximize $J$ on $\mathcal{Q}$ is that the Gateaux (directional) derivative of $J$ at $q^*$, taken in the direction of any other feasible mechanism $p \in \mathcal{Q}$, must be non-positive. Moving from $q^*$ towards $p$ must not increase the objective. Thus, for all $p \in \mathcal{Q}$:
\begin{equation*}
    dJ(q^*; p - q^*) \le 0.
\end{equation*}
We compute this Gateaux derivative using the chain rule:
\begin{equation*}
    dJ(q^*; p - q^*) = \frac{\partial J}{\partial m_{P_1}^{q^*}} \int_{\mathcal{X}} (p(x) - q^*(x)) dP_1(x) + \frac{\partial J}{\partial m_{P_0}^{q^*}} \int_{\mathcal{X}} (p(x) - q^*(x)) dP_0(x).
\end{equation*}
Let $C_1(q^*)$ and $C_0(q^*)$ denote the partial derivatives of the KL-divergence with respect to the marginals:
\begin{align*}
   C_1(q^*) &= \log\left(\frac{m_{P_1}^{q^*}(1-m_{P_0}^{q^*})}{m_{P_0}^{q^*}(1-m_{P_1}^{q^*})}\right)  \\
   C_0(q^*) &= \frac{m_{P_0}^{q^*} - m_{P_1}^{q^*}}{m_{P_0}^{q^*}(1-m_{P_0}^{q^*})}.
\end{align*}
By substituting the Radon-Nikodym derivative $L(x) = \frac{dP_1}{dP_0}(x)$, we can factor the directional derivative into a single integral with respect to the null measure $P_0$:
\begin{equation*}
   dJ(q^*; p - q^*) = \int_{\mathcal{X}} (p(x) - q^*(x)) \left[ C_1(q^*) L(x) + C_0(q^*) \right] dP_0(x) \le 0.
\end{equation*}
\vspace{1em}
\textbf{Step 3:} Define $K(x) = C_1(q^*) L(x) + C_0(q^*)$. The optimality condition requires that for \textit{every} feasible alternative mechanism $p \in \mathcal{Q}$:
\begin{equation*}
    \int_{\mathcal{X}} (p(x) - q^*(x)) K(x) dP_0(x) \le 0.
\end{equation*}
To guarantee this integral is non-positive for all possible choices of $p(x) \in [y_{min}, y_{max}]$, the mechanism $q^*(x)$ must systematically pointwise neutralize the sign of the $K(x)$. We analyze this in two cases:
\begin{itemize}
    \item {When $K(x) > 0$:} The product $(p(x) - q^*(x))K(x)$ must be $\le 0$. Since $K(x)$ is positive, we must have $p(x) - q^*(x) \le 0 \implies p(x) \le q^*(x)$ for all $p \in \mathcal{Q}$. Because $p(x)$ can be as large as $y_{max}$, the only way to satisfy this bound for all $p$ is if we set $q^*(x)$ to its upper bound $y_{max}$.
    \item {When $K(x) < 0$:} The product $(p(x) - q^*(x))K(x)$ must again be $\le 0$. Since $K(x)$ is negative, we must have $p(x) - q^*(x) \ge 0 \implies p(x) \ge q^*(x)$ for all $p \in \mathcal{Q}$. Because $p(x)$ can be as small as $y_{min}$, the only way to satisfy this is if we set $q^*(x)$ to its lower bound $y_{min}$.
\end{itemize}
Since $m_{P_1}^{q^*} > m_{P_0}^{q^*}$ for any non-trivial mechanism, the log-odds ratio $C_1(q^*)$ is strictly positive. Therefore, we can algebraically rearrange the condition $K(x) > 0$ to isolate the likelihood ratio:
\begin{equation*}
   C_1(q^*) L(x) + C_0(q^*) > 0 \iff L(x) > -\frac{C_0(q^*)}{C_1(q^*)} .
\end{equation*}
Thus, the optimal mechanism is:
\begin{equation*}
   Q(1|x) = y_{max} \cdot \mathbbm{1}\left(L(x) > t\right) + y_{min} \cdot \mathbbm{1}\left(L(x) \le t\right),
\end{equation*}
where the explicit optimal threshold $t$ is strictly given by the ratio of the gradients:
\begin{equation*}
   t = -\frac{C_0(q^*)}{C_1(q^*)} = \frac{m_{P_1}^{q^*} - m_{P_0}^{q^*}}{m_{P_0}^{q^*}(1-m_{P_0}^{q^*}) \log\left(\frac{m_{P_1}^{q^*}(1-m_{P_0}^{q^*})}{m_{P_0}^{q^*}(1-m_{P_1}^{q^*})}\right)}.
\end{equation*}
\end{proof}

\begin{proof}[Proof of \Cref{thm:bin-opt-comp}]
By the definition of the infimum, for any fixed mechanism $Q$, the worst-case KL divergence over the composite classes is bounded above by the KL divergence at the specific pair $(P_0^*, P_1^*)$:
    \begin{equation}
        \inf_{\substack{P_1 \in \mathcal{P}_1 \\ P_0 \in \mathcal{P}_0}} D_{KL}(M_{P_1}^{Q} \| M_{P_0}^{Q}) \le D_{KL}(M_{P_1^*}^{Q} \| M_{P_0^*}^{Q}).
    \end{equation}
    Taking the supremum over all admissible mechanisms $Q \in \mathcal{D}_{\epsilon}^{(2)}$ on both sides yields:
    \begin{equation} \label{eq:upper_bound}
        \sup_{Q \in \mathcal{D}_{\epsilon}^{(2)}} \inf_{\substack{P_1 \in \mathcal{P}_1 \\ P_0 \in \mathcal{P}_0}} D_{KL}(M_{P_1}^{Q} \| M_{P_0}^{Q}) \le \sup_{Q \in \mathcal{D}_{\epsilon}^{(2)}} D_{KL}(M_{P_1^*}^{Q} \| M_{P_0^*}^{Q}) = D_{KL}(M_{P_1^*}^{Q^*} \| M_{P_0^*}^{Q^*}),
    \end{equation}
    where the equality follows from the definition of $Q^*$ and \Cref{thm:opt-binary-simple}.

     We now show that for the fixed optimal mechanism $Q^*$, the LFD pair $(P_0^*, P_1^*)$ minimizes the KL divergence. Let $L^*(x) = \frac{dP_1^*(x)}{dP_0^*(x)}$ be the likelihood ratio of the LFD pair.
    Recall the structure of the optimal binary mechanism $Q^*$:
    \begin{equation}
        Q^*(1|x) = \begin{cases} 
            \frac{e^\epsilon}{1+e^\epsilon} & \text{if } L^*(x) > t^*, \\ 
            \frac{1}{1+e^\epsilon} & \text{if } L^*(x) \le t^*. \end{cases}
    \end{equation}
    By inspection, $Q^*(1|x)$ is a monotonically non-decreasing function of the likelihood ratio $L^*(x)$. Let us denote this function as $\psi(L^*(x)) = Q^*(1|x)$.
    
    By the definition of a least favorable distribution pair, the likelihood ratio $L^*(X)$ exhibits stochastic dominance over the composite classes. Specifically, for any non-decreasing function $f$:
    \begin{align}
        \mathbb{E}_{P_1}[f(L^*(X))] &\ge \mathbb{E}_{P_1^*}[f(L^*(X))] \quad \forall P_1 \in \mathcal{P}_1 \label{eq:stoch_dom_1} \\
        \mathbb{E}_{P_0}[f(L^*(X))] &\le \mathbb{E}_{P_0^*}[f(L^*(X))] \quad \forall P_0 \in \mathcal{P}_0. \label{eq:stoch_dom_0}
    \end{align}
    Let $m_{P}^{Q^*} = \mathbb{E}_{P}[Q^*(1|X)]$ be the marginal probability of outputting $1$ under distribution $P$. Substituting our non-decreasing function $\psi$ into equations \eqref{eq:stoch_dom_1} and \eqref{eq:stoch_dom_0}, we obtain:
    \begin{align}
    \label{ineq-mean}
        m_{P_1}^{Q^*} \ge m_{P_1^*}^{Q^*} \quad \forall P_1 \in \mathcal{P}_1 ,~~\text{and}
~~~        m_{P_0}^{Q^*} \le m_{P_0^*}^{Q^*} \quad \forall P_0 \in \mathcal{P}_0.
    \end{align}
    The induced marginals $M_P^{Q^*}$ are Bernoulli distributions. The KL divergence between two Bernoulli distributions with success probabilities $p$ and $q$ is defined as $d(p, q) = p \log(p/q) + (1-p) \log((1-p)/(1-q))$. 
    For $p > q$, the function $d(p, q)$ is strictly increasing in $p$ and strictly decreasing in $q$. Since $Q^*$ optimally distinguishes $P_1^*$ and $P_0^*$, we strictly have $m_{P_1^*}^{Q^*} > m_{P_0^*}^{Q^*}$.
    
    Therefore, $m_{P_1}^{Q^*} \ge m_{P_1^*}^{Q^*} > m_{P_0^*}^{Q^*} \ge m_{P_0}^{Q^*}$ implies that the KL divergence can only increase for any other $P_1 \in \mathcal{P}_1$ and $P_0 \in \mathcal{P}_0$:
    \begin{equation}
        D_{KL}(M_{P_1}^{Q^*} \| M_{P_0}^{Q^*}) = d(m_{P_1}^{Q^*}, m_{P_0}^{Q^*}) \ge d(m_{P_1^*}^{Q^*}, m_{P_0^*}^{Q^*}) = D_{KL}(M_{P_1^*}^{Q^*} \| M_{P_0^*}^{Q^*})
    \end{equation}
    Taking the infimum over $P_1 \in \mathcal{P}_1, P_0 \in \mathcal{P}_0$, we get:
    \begin{equation} \label{eq:inf_bound}
        \inf_{\substack{P_1 \in \mathcal{P}_1 \\ P_0 \in \mathcal{P}_0}} D_{KL}(M_{P_1}^{Q^*} \| M_{P_0}^{Q^*}) \ge D_{KL}(M_{P_1^*}^{Q^*} \| M_{P_0^*}^{Q^*})
    \end{equation}

     Since the supremum over all mechanisms $\mathcal{D}_{\epsilon}^{(2)}$ must be at least the value achieved by the specific mechanism $Q^*$, we have:
    \begin{equation}
        \sup_{Q \in \mathcal{D}_{\epsilon}^{(2)}} \inf_{\substack{P_1 \in \mathcal{P}_1 \\ P_0 \in \mathcal{P}_0}} D_{KL}(M_{P_1}^{Q} \| M_{P_0}^{Q}) \ge \inf_{\substack{P_1 \in \mathcal{P}_1 \\ P_0 \in \mathcal{P}_0}} D_{KL}(M_{P_1}^{Q^*} \| M_{P_0}^{Q^*})
    \end{equation}
    Applying the lower bound from \eqref{eq:inf_bound}:
    \begin{equation} \label{eq:lower_bound_final}
        \sup_{Q \in \mathcal{D}_{\epsilon}^{(2)}} \inf_{\substack{P_1 \in \mathcal{P}_1 \\ P_0 \in \mathcal{P}_0}} D_{KL}(M_{P_1}^{Q} \| M_{P_0}^{Q}) \ge D_{KL}(M_{P_1^*}^{Q^*} \| M_{P_0^*}^{Q^*})
    \end{equation}
    Combining the upper bound \eqref{eq:upper_bound} and the lower bound \eqref{eq:lower_bound_final}, we conclude that the inequalities are in fact equalities, completing the first part.

   Now, $\mathcal{P}_0^{Q^*}=\{\text{Ber}(m_{P_0}):P_0\in\mathcal{P}_0\}$ vs. $\mathcal{P}_1^{Q^*}=\{\text{Ber}(m_{P_1}):P_1\in\mathcal{P}_1\},$ and in \eqref{ineq-mean}, we derived that $m_{P_0^*}=\max\{m_{P_0}:P_0\in\mathcal{P}_0\}<\min\{m_{P_1}:P_1\in\mathcal{P}_1\}=m_{P_1^*}.$ Thus, $(M_{P_0^*}^{Q^*},M_{P_1^*}^{Q^*})$ is the LFD pair for $\mathcal{P}_0^{Q^*}$ vs. $\mathcal{P}_1^{Q^*}$.  Now, it follows from Theorem 2.1 of \cite{saha2025huber} that is an e-variable under $\mathcal{P}_0^{Q^*}=\{M_{P_0}^{Q^*}:P_0\in \mathcal{P}_0\}$ and
    $$ \sup_{E\in \mathcal{E}(\mathcal{P}_0^{Q^*})} \inf_{M\in\mathcal{P}_1^{Q^*}}\mathbb{E}_{M}\!\left(\log E(Y)\right)=\inf_{M\in\mathcal{P}_1^{Q^*}}\mathbb{E}_{M}\!\left(\log E^*(Y)\right)=\kl(M_{P_1^*}^{\Q^*},M_{P_0^*}^{\Q^*}).$$
This completes the proof of the theorem.  
\end{proof}

\subsection{Omitted proofs from Section 3}
\label{proof-sec3}
\begin{proof}[Proof of \Cref{thm:quantization}]
Without loss of generality, assume $u_1 > u_0$. Let $S$ denote the event that $E = u_1.$ So, any binary e-variable $E$ can be written as
\[ E = u_1 \mathbb{I}_S + u_0 \mathbb{I}_{S^c}. \]

\textbf{Step 1:}
We fix $u_1 > u_0$ and substitute this form into the optimization problem for $S$. The objective function becomes:
  \begin{align*}
\mathbb{E}_{P_1}[\log(E)] &= P_1(S) \log(u_1) + (1 - P_1(S)) \log(u_0) \\
&= P_1(S) (\log(u_1) - \log(u_0)) + \log(u_0).
\end{align*}
Let $\beta = P_1(S)$. Since we assumed $u_1 > u_0$, the term $(\log(u_1) - \log(u_0))$ is strictly positive. Therefore, the objective function is strictly increasing in the power $\beta$. To maximize the expected log-growth, we must maximize $\beta$.

Next, consider the constraint $\mathbb{E}_{P_0}[E] \le 1$:
  \begin{align*}
P_0(S) u_1 + (1 - P_0(S)) u_0 &\le 1 \\
P_0(S) (u_1 - u_0) &\le 1 - u_0.
\end{align*}
Let $\alpha = P_0(S)$. Since $u_1 > u_0$, we can rearrange the inequality to:
  \[ \alpha \le \frac{1 - u_0}{u_1 - u_0}. \]
Let $\alpha_{\max} = \frac{1 - u_0}{u_1 - u_0}$. The problem of optimizing the set $S$ for fixed values $u_1, u_0$ reduces to:
  \[
    \displaystyle\text{Maximize}_S~~~ P_1(S) \quad \text{subject to } P_0(S) \le \alpha_{\max}.
    \]
The Neyman-Pearson Lemma states that the unique solution to this problem (maximizing power for a bounded size) is given by the likelihood ratio test. Specifically, there exists a threshold $t$ such that the optimal set is a level set of the likelihood ratio $L = dP_1/dP_0$:
  \[ S^* = [ L > t].\]
Consequently, the optimal estimator $E^*$ must be of the form:
  \[ E^* = \begin{cases} u_1 & \text{if } L(X) > t,\\ u_0 & \text{if } L(X) \le t. \end{cases} \]
This justifies restricting the search space to threshold functions of $L$.

\textbf{Step 2: }
Fix an arbitrary measurable set $S \in \mathcal{F}$. This fixes the probability masses:
  \[ \alpha = P_0(S), \quad \beta = P_1(S).\]
Note that these probabilities depend only on the geometry of the set $S$, not on the scalar values $u_1, u_0$ assigned to it.
We now solve the inner maximization problem: find the optimal scalars $u_1, u_0$ for this fixed set $S$.
\begin{align*}
\text{Maximize} \quad & \beta \log(u_1) + (1-\beta) \log(u_0), \\
\text{Subject to} \quad & \alpha u_1 + (1-\alpha) u_0 = 1.
\end{align*}

We form the Lagrangian with multiplier $\lambda$:
  \[ \mathcal{L} = \beta \log(u_1) + (1-\beta) \log(u_0) - \lambda (\alpha u_1 + (1-\alpha) u_0 - 1). \]
Solving the first-order conditions $\frac{\partial \mathcal{L}}{\partial u} = 0$ yields:
  \[ u_1^* = \frac{\beta}{\lambda \alpha}, \quad u_0^* = \frac{1-\beta}{\lambda (1-\alpha)}.\]
Substituting into the constraint forces $\lambda = 1$. Thus, for any fixed set $S$, the conditionally optimal values are uniquely determined as:
  \[ u_1^*(S) = \frac{\beta}{\alpha}, \quad u_0^*(S) = \frac{1-\beta}{1-\alpha}.\]
Since we only focus on $S$ of the form $[L>t]$, we can rewrite $u_1,u_0$ as functions of $t$:
  \[ u_1^*(t) = \frac{\beta(t)}{\alpha(t)}, \quad u_0^*(t) = \frac{1-\beta(t)}{1-\alpha(t)},\]
where $\alpha(t)=P_0[L>t]$ and $\beta(t)=P_1[L>t]$.
This result holds for any choice of $S$. Therefore, we can substitute these optimal values back into the original objective function, reducing the problem to a single optimization over the parameter $t$.

\textbf{{Step 3: }}
Having established the optimal values $u_1(t)$ and $u_0(t)$ for any fixed $t$, we substitute them into the objective function. The objective becomes a function of a single scalar variable $t$: 
  \begin{equation*}
J(t) = \beta(t)\log\left(\frac{\beta(t)}{\alpha(t)}\right) + (1-\beta(t))\log\left(\frac{1-\beta(t)}{1-\alpha(t)}\right).
\end{equation*} 

Let $Y = L(X)$ denote the likelihood ratio as a scalar random variable taking values in $\mathbb{R}^{+}$. Let $Q_0$ and $Q_1$ be the pushforward probability measures of $Y$ under the null and alternative hypotheses $P_0$ and $P_1$, respectively. We first rigorously establish the relationship between these measures without assuming the existence of a continuous Lebesgue density.

By the definition of the Radon-Nikodym derivative, $dP_1 = L \, dP_0$. For any Borel set $B \subseteq \mathbb{R}^{+}$, the probability measure of $Y$ under the alternative hypothesis is given by:
  \begin{equation*}
Q_1(B) = P_1(Y \in B) = \int_{\{x: L(x) \in B\}} dP_1(x).
\end{equation*} 

Substituting the Radon-Nikodym derivative into the integral yields:
  \begin{equation*}
Q_1(B) = \int_{\{x: L(x) \in B\}} L(x) \, dP_0(x).
\end{equation*} 

Applying the change of variables theorem for pushforward measures (the Law of the Unconscious Statistician), we can integrate directly over the space of $Y$ using its null measure $Q_0$:
  \begin{equation*}
Q_1(B) = \int_{B} y \, dQ_0(y).
\end{equation*} 
Because this equality holds for every measurable set $B$, the measures satisfy the relationship 
\begin{equation}
\label{eq:changeof-var}
dQ_1(y) = y \, dQ_0(y) \text{ almost everywhere.
}
\end{equation}
Applying the multivariate chain rule to $J(t) \equiv f(\alpha(t), \beta(t))$ with respect to the underlying probability measure yields the Lebesgue-Stieltjes differential:
  \begin{equation*}
dJ(t) = \frac{\partial f}{\partial \alpha} d\alpha(t) + \frac{\partial f}{\partial \beta} d\beta(t).
\end{equation*}
First, we differentiate the objective function and compute the partial derivatives evaluated at the optimal scalars $u_1(t) = \frac{\beta(t)}{\alpha(t)}$ and $u_0(t) = \frac{1-\beta(t)}{1-\alpha(t)}$: 
  \begin{align*}
\frac{\partial f}{\partial \alpha} &= -\frac{\beta(t)}{\alpha(t)} + \frac{1-\beta(t)}{1-\alpha(t)} = -u_1(t) + u_0(t) \\
\frac{\partial f}{\partial \beta} &= \log u_1(t) - \log u_0(t).
\end{align*}
As $t$ decreases, the differential probability mass added to the rejection region is given by the non-negative measure $d(-\alpha(t)) = dQ_0(t)$. And \eqref{eq:changeof-var} imposes that $d(-\beta(t)) = dQ_1(t) = t \, dQ_0(t)$. Evaluating the total differential of $J$ with respect to the marginal expansion of the region yields: 
  \begin{equation*}
dJ(t) = \left[ (-u_1(t) + u_0(t)) + t(\log u_1(t) - \log u_0(t)) \right] dQ_0(t).
\end{equation*} 
Since $dQ_0(t)$ is a strictly positive measure, the objective $J(t)$ strictly increases as $t$ decreases as long as the bracketed term is positive. We define:
  \begin{equation*}
H(t) = (u_1(t) - u_0(t)) - t(\log u_1(t) - \log u_0(t)).
\end{equation*}
We get an optimal threshold implies setting $H(t^*) = 0$: 
  \begin{equation*}
-t^*(\log u_1(t^*) - \log u_0(t^*)) + (u_1(t^*) - u_0(t^*)) = 0.
\end{equation*} 
Solving for $t^*$: 
  \begin{equation*}
t^* = \frac{u_1(t^*) - u_0(t^*)}{\log u_1(t^*) - \log u_0(t^*)}.
\end{equation*} 
Note that this optimal threshold may not be unique when the distribution of $L$ under the null is not continuous. However, the threshold $t^*$ always remains a valid, globally optimal choice (although it may not be unique). 
\end{proof}

\subsection{Omitted proofs from Section 4}
\label{proof-sec4}
\begin{proof}[Proof of \Cref{thm:bounded}]
We prove a more general version of \Cref{thm:bounded} using a general utility function $U:(0,\infty)\mapsto\mathbb {R}$, which is concave and strictly increasing. Then, \Cref{thm:bounded} follows by plugging in $U=\log$ below.

Let the domain of feasible random variables be $\mathcal{D} = \{E : c_1\leq E\leq c_2 ~~~P_0-\text{a.s. }, \mathbb{E}_{P_0}[E] \le 1\}$. 
We seek a random variable $E^*$ that solves the following constrained optimization problem:
\begin{align*}
\sup_{E \in \mathcal{D}} \quad  \mathbb{E}_{P_1}[U(E)] \quad
\text{subject to} \quad & \mathbb{E}_{P_0}[E] \le 1.
\end{align*}
\begin{theorem}
\label{thm:bounded-gen}
The solution $E^*$ to the optimization problem is given by:
\[
E^* = \min\left( c_2, \max\left( c_1, g(L(X)) \right) \right),
\]
where $L(X)=\frac{dP_1(X)}{dP_0(X)}$ is the likelihood ratio and $g$ is some strictly increasing function.
\end{theorem}


The space $\mathcal{D}$ is a closed, bounded, and convex subset of $L^\infty(P_0)$. By Banach-Alaoglu theorem, it is weak-* compact. Because $U$ is concave and continuous, the objective functional $J(E)$ is weak-* upper semicontinuous. By the extreme value theorem for weak-* topologies, a global maximum $E^* \in \mathcal{D}$ is guaranteed to exist.

To prove uniqueness, suppose there exist two optimal solutions $E_1, E_2 \in \mathcal{D}$ such that $P_0(E_1 \neq E_2) > 0$. Because $\mathcal{D}$ is convex, the midpoint $E_{mid} = \frac{1}{2}E_1 + \frac{1}{2}E_2$ is strictly feasible ($E_{mid} \in \mathcal{D}$). 
Because $U$ is strictly concave, Jensen's inequality implies:
\begin{equation}
    U\left(\frac{1}{2}E_1(X) + \frac{1}{2}E_2(X)\right) > \frac{1}{2}U(E_1(X)) + \frac{1}{2}U(E_2(X)) \quad \text{on the event} \{E_1 \neq E_2\}.
\end{equation}
Taking the expectation under $P_1$, we obtain $J(E_{mid}) > \frac{1}{2}J(E_1) + \frac{1}{2}J(E_2) = J(E^*)$. This contradicts the optimality of $E_1$ and $E_2$. Thus, the optimal solution $E^*$ must be unique $P_0$-almost surely.

We rewrite the objective function in terms of $P_0$: $\mathbb{E}_{P_1}[U(E)] = \mathbb{E}_{P_0}[L U(E)]$.
We introduce a Lagrange multiplier $\lambda \ge 0$ for the integral constraint $\mathbb{E}_{P_0}[E] \le 1$. We restrict our search to functions within the domain $\mathcal{D}$.
The Lagrangian $\mathcal{L}(E, \lambda)$ is defined as:
\[
\mathcal{L}(E, \lambda) = \mathbb{E}_{P_0}[L U(E)] - \lambda \left( \mathbb{E}_{P_0}[E] - 1 \right)
 = \lambda + \mathbb{E}_{P_0} \left[ LU(E) - \lambda E \right].
\]
To maximize $\mathcal{L}(E, \lambda)$ over $E \in \mathcal{D}$ for a fixed $\lambda > 0$, we maximize the integrand pointwise for each $X \in X$. For any realised values $E=y$ and $L(X)=\ell$, we solve the scalar optimization problem:
\begin{equation}
\label{eq:opti-y}
    \max_{y} \quad \phi(y) = \ell U(y) - \lambda y \quad \text{subject to } y \in [c_1, c_2].
\end{equation}

Since $U$ is concave and strictly increasing, the function $\phi$ is strictly concave. So, there exists a unique (unconstrained) maximizer of $\phi(y)$. Let, for each fixed $g$, the unconstrained global maximum occur at $y=g(l)$. Since we are maximizing a concave function over a closed interval $[c_1, c_2]$, the solution is the projection of the unconstrained maximum onto the interval. This results in three cases:
\begin{enumerate}
    \item If $g(l) < c_1$, the function is decreasing on the interval; maximum is at $c_1$.
    \item If $g(l) > c_2$, the function is increasing on the interval; maximum is at $c_2$.
    \item Otherwise, the maximum is at the interior point $g(l)$.
\end{enumerate}
Thus, the optimal solution for \eqref{eq:opti-y} is $y=\min\left( c_2, \max\left( c_1, g(l) \right) \right)$. Therefore, the optimal solution to the original optimization problem, for a fixed $\lambda$, is:
\[
E_\lambda = \min\left( c_2, \max\left( c_1, g(L(X)) \right) \right).
\]
We establish monotonicity using pure algebraic inequalities derived from the definition of optimality. Let $\ell_1$ and $\ell_2$ be two likelihood ratios such that $0 < \ell_1 < \ell_2$. Let $y_1 = g(\ell_1)$ and $y_2 = g(\ell_2)$ be their corresponding unconstrained maximizers.

By the definition of $y_1$ as the unique maximizer for $\ell_1$, evaluating the objective function at $y_1$ must yield a value greater than evaluating it at $y_2$:
\begin{equation}
    \ell_1 U(y_1) - \lambda y_1 > \ell_1 U(y_2) - \lambda y_2.
    \label{eq:opt1}
\end{equation}
Similarly, by the definition of $y_2$ as the unique maximizer for $\ell_2$, evaluating the objective function at $y_2$ must yield a value greater than or equal to evaluating it at $y_1$:
\begin{equation}
    \ell_2 U(y_2) - \lambda y_2> \ell_2 U(y_1) - \lambda y_1.
    \label{eq:opt2}
\end{equation}

We can rearrange (\ref{eq:opt1}) toget:
\begin{equation}
    \lambda(y_1 - y_2)<\ell_1 (U(y_1) - U(y_2)).
    \label{eq:rearrange1}
\end{equation}
We can rearrange (\ref{eq:opt2}) to isolate the penalty terms in the exact same direction:
\begin{equation}
    \ell_2 (U(y_1) - U(y_2))<\lambda(y_1 - y_2).
    \label{eq:rearrange2}
\end{equation}

By chaining the inequalities (\ref{eq:rearrange1}) and (\ref{eq:rearrange2}) together, the $\lambda$ terms are eliminated, yielding:
\begin{equation}
    \ell_2 (U(y_1) - U(y_2)) \le \ell_1 (U(y_1) - U(y_2)).
\end{equation}
Rearranging this to group by the utility differences gives:
\begin{equation}
    (\ell_2 - \ell_1) (U(y_1) - U(y_2))< 0.
\end{equation}

By our initial assumption, $\ell_2 > \ell_1$, meaning the term $(\ell_2 - \ell_1)$ is strictly positive. Therefore, for the product to be non-positive, we must have:
\begin{equation}
    U(y_1) - U(y_2) < 0 \implies U(y_1) < U(y_2).
\end{equation}
Because the utility function $U$ is strictly increasing, $U(y_1) < U(y_2)$ directly implies:
\begin{equation}
    y_1 < y_2.
\end{equation}
Since $\ell_1 < \ell_2 \implies g(\ell_1)< g(\ell_2)$, the function $g(\ell)$ istrictly increasing.

Now, since the utility function $U(\cdot)$ is strictly increasing, the optimal solution must saturate the budget constraint:
\[
\mathbb{E}_{P_0}[E_{\lambda^*}] = 1.
\]
The random variable $E^*$ constructed with $\lambda^*$ maximizes the Lagrangian over the domain $\mathcal{D}$ and satisfies all constraints. By the Lagrange Sufficiency Theorem, $E^*$ is the global maximizer.

For $U$ is $\log,$ simple calculation shows $g(l)=\frac{\ell}{\lambda}$ and hence, 
\[
E^* = \min\left( c_2, \max\left( c_1,\frac{L(X)}{\lambda} \right) \right).
\]

\end{proof}

\subsection{Omitted proofs from Section 5}
\label{proof-sec5}
\begin{proof}[Proof of \Cref{thm:convex-constraint}]
We prove a more general version of \Cref{thm:convex-constraint} using a general utility function $U:(0,\infty)\mapsto\mathbb {R}$, which is concave and strictly increasing. Then we show that \Cref{thm:convex-constraint} follows by plugging in $U=\log$ below.

We consider the optimization problem
\begin{equation}
\begin{aligned}
& \underset{E}{\text{maximize}}
& & \mathbb{E}_{P_1}[U(E)] \\
& \text{subject to}
& & \mathbb{E}_{P_0}[E] \le 1, \quad \mathbb{E}_{P_0}[\phi(E)] \le C.
\end{aligned}
\end{equation}
\begin{theorem}
\label{thm:convex-constraint-gen}
 Let  $U:(0,\infty)\mapsto\mathbb {R}$ be a concave and strictly increasing function and $\phi : (0, \infty) \to \mathbb{R}$ be strictly convex and superlinear: $\lim_{x \to \infty} \frac{\phi(x)}{x} = \infty.$
 Then for the above optimization problem,
\begin{itemize}
    \item[(i)] There exists a unique maximizer $E^*$ (up to $P_0$-a.s.).
    \item[(ii)] 
    There exists an increasing function $\psi : (0, \infty) \to (0, \infty)$ such that $E^* = \psi(L) \quad P_0\text{-a.s.}$
\end{itemize} 
\end{theorem}

\textbf{Part (i):}
Let $\mathcal{C}$ denote the feasible set of random variables:
\[ \mathcal{C} = \{ E \in L^1(P_0) : E \ge 0 \ P_0\text{-a.s.}, \ \mathbb{E}_{P_0}[E] \le 1, \ \mathbb{E}_{P_0}[\phi(E)] \le C \}. \]
Because the expectation operator is linear and $\phi$ is convex, $\mathcal{C}$ is a convex subset of $L^1(P_0)$. The objective functional $J(E) = \mathbb{E}_{P_1}[U(E)] = \mathbb{E}_{P_0}[L U(E)]$ is strictly concave because the logarithm is strictly concave. The supremum of a strictly concave functional over a convex set is achieved by at most one point. Thus, if a maximizer $E^*$ exists, it is unique up to $P_0$-almost sure equivalence.

To prove existence, we use the superlinearity condition $\lim_{e \to \infty} \phi(x)/x = \infty$ combined with the bounded integral $\mathbb{E}_{P_0}[\phi(E)] \le C$, which implies, by the de la Vall\'ee-Poussin theorem, that the family of random variables $\mathcal{C}$ is uniformly integrable. By the Dunford-Pettis theorem, uniform integrability ensures that $\mathcal{C}$ is relatively weakly compact in $L^1(P_0)$. Since $\mathcal{C}$ is also convex and strongly closed (which follows from Fatou's Lemma), it is weakly closed. Therefore, $\mathcal{C}$ is weakly compact in $L^1(P_0)$.

The functional $J(E) = \mathbb{E}_{P_0}[L U(E)]$ is upper semi-continuous with respect to the weak topology on $L^1(P_0)$ (again, by Fatou's Lemma and the concavity of the $U$). An upper semi-continuous functional defined on a weakly compact set achieves its supremum. Thus, a global maximizer $E^* \in \mathcal{C}$ exists.

\vspace{1em}
\textbf{Part (ii):} By the generalized Karush-Kuhn-Tucker (KKT) theorem (or the method of Lagrange multipliers for infinite-dimensional spaces), there exist scalar multipliers $\lambda \ge 0$ and $\gamma \ge 0$ such that $E^*$ maximizes the unconstrained Lagrangian:
\[ \mathcal{L}(E) = \mathbb{E}_{P_0}[L U(E)] - \lambda (\mathbb{E}_{P_0}[E] - 1) - \gamma (\mathbb{E}_{P_0}[\phi(E)] - C). \]
We can rewrite the Lagrangian by grouping the terms inside the expectation:
\[ \mathcal{L}(E) = \mathbb{E}_{P_0} \left[ L U( E) - \lambda E - \gamma \phi(E) \right] + \lambda + \gamma C. \]
To maximize this functional, we can optimize the expression inside the expectation pointwise $P_0$-almost surely. For any realization where $L = l$ and we choose $E = x$, the pointwise optimization problem is:
\[ \max_{x > 0} f_l(x), \quad \text{where } f_l(x)=l U(x) - \lambda x - \gamma \phi(x). \]
Since $\phi$ is strictly convex and $U$ is strictly concave, $f_l$ is a strictly concave function, and hence the above optimization problem has a unique solution at $x=\psi(l)$. For $l_1>l_2$, we need to show $\psi(l_1)>\psi(l_2)$. 

Assume $l_1 > l_2$. Let $x_1 = \psi(l_1)$ be the unique maximizer for $l_1$ and $x_2 = \psi(l_2)$ be the unique maximizer for $l_2$. By the definition of a unique maximizer, we know:$f_{l_1}(x_1) > f_{l_1}(x_2)$ and $f_{l_2}(x_2) > f_{l_2}(x_1)$.
Expanding the two inequalities above using the definition of $f_l(x)$:

(1) $l_1 U(x_1) - \lambda x_1 - \gamma \phi(x_1) > l_1 U(x_2) - \lambda x_2 - \gamma \phi(x_2)$

(2) $l_2 U(x_2) - \lambda x_2 - \gamma \phi(x_2) > l_2 U(x_1) - \lambda x_1 - \gamma \phi(x_1)$

Now, let's rearrange both to group the $U(x)$ terms on one side and the terms ($\lambda x + \gamma \phi(x)$) on the other:

(1) $l_1 [U(x_1) - U(x_2)] > [\lambda x_1 + \gamma \phi(x_1)] - [\lambda x_2 + \gamma \phi(x_2)]$

(2) $l_2 [U(x_2) - U(x_1)] > [\lambda x_2 + \gamma \phi(x_2)] - [\lambda x_1 + \gamma \phi(x_1)]$

Multiply inequality (2) by $-1$ (which flips the inequality sign) to make the right-hand side match inequality (1):$$-l_2 [U(x_2) - U(x_1)] < -([\lambda x_2 + \gamma \phi(x_2)] - [\lambda x_1 + \gamma \phi(x_1)]).$$
This implies:$$0< (l_1-l_2) [U(x_1) - U(x_2)].$$
Since $l_1>l_2$, we must have $U(x_1) >U(x_2)$. Now, $U$ is strictly increasing. Therefore, $x_1>x_2$, that is $\psi(\ell)$ a strictly increasing function of $\ell$.
Therefore, the opmila e-value for the original problem is 
\[E^*=\psi(L), \text{ for some strictly increasing function } \psi.\]
\paragraph{Special case when $U$ is $\log$:} The unconstrained maximum is found by setting the first derivative with respect to $x$ to zero:
\[ \frac{l}{x} - \lambda - \gamma \phi'(x) = 0. \]
Rearranging this gives the required first-order condition for the random variables:
\[ \frac{L}{E^*} = \lambda + \gamma \phi'(E^*) \quad P_0\text{-a.s.} \]
From the first-order condition, we can express the likelihood ratio $L$ as a function of $E^*$ $P_0$-almost surely:
\[ L = E^* \left( \lambda + \gamma \phi'(E^*) \right) \quad P_0\text{-a.s.} \]
\end{proof}

\subsection{Omitted proofs from Section 6}
\label{proof-sec6}
\begin{proof}[Proof of \Cref{thm:comp-gen}]
We prove a more general version of \Cref{thm:comp-gen} using a general utility function $U:(0,\infty)\mapsto\mathbb {R}$, which is concave and strictly increasing. The prrof of \Cref{thm:comp-gen} follows by plugging in $U=\log$ below.

\begin{theorem}
\label{thm:gen} Suppose that $U:(0,\infty)\mapsto\mathbb {R}$ is a concave and strictly increasing function.
Let $E^*$ be the optimal solution for the simple hypothesis pair $(P_0^*, P_1^*)$, defined as:
\[
E^* = \argmax_{E \in \mathcal{E}^\prime(\{P_0^*\})}  \mathbb{E}_{P_1^*}[U(E)].
\]
Assume that it is of the form $E^*=\psi(L^*),$ for some non-decreasing function $\psi$.
Then $E^*\in\mathcal{E}^\prime(\mathcal{P}_0)$ and
\begin{align*}
    \sup_{E \in \mathcal{E}^\prime(\mathcal{P}_0)} \inf_{P_1 \in \mathcal{P}_1} \mathbb{E}_{P_1}[U(E)] = &\inf_{P_1 \in \mathcal{P}_1} \mathbb{E}_{P_1}[U(E^*)]\\
    &= \mathbb{E}_{P_1^*}[U(E^*)].
\end{align*}
\end{theorem}

The proof proceeds in three steps: establishing validity, determining worst-case performance, and proving optimality.

We first show that $E^* \in \mathcal{E}^\prime(\mathcal{P}_0)$. 
$\psi(x)$ is a {non-decreasing} function of $x$.
Since for all $P_0 \in \mathcal{P}_0$, $L^*$ under $P_0^*$ stochastically dominates $L^*$ under $P_0$, we have
\[
  \mathbb{E}_{P_0}[E^*] = \mathbb{E}_{P_0}[\psi(L^*)] \le \mathbb{E}_{P_0^*}[\psi(L^*)].
  \]
By the definitiopn, we know $\mathbb{E}_{P_0^*}[E^*] \leq 1$. Therefore:
  \[
    \mathbb{E}_{P_0}[E^*] \le 1 \quad \forall P_0 \in \mathcal{P}_0.
    \]
Thus, $E^*$ is a valid e-variable for the composite null.

Next, we evaluate the worst-case growth rate of $E^*$ under the alternative $\mathcal{P}_1$.
Define the function $h(x) = U(\psi(x))$. Since $U(\cdot)$ is strictly increasing and $\psi(\cdot)$ is non-decreasing, the composite function $h(x)$ is {non-decreasing}.
Since for all $P_1 \in \mathcal{P}_1$, $L^*$ under $P_1$ stochastically dominates $L^*$ under $P_1^*$, we have:
  \[
    \mathbb{E}_{P_1}[U(E^*)] = \mathbb{E}_{P_1}[h(L^*)] \ge \mathbb{E}_{P_1^*}[h(L^*)] = \mathbb{E}_{P_1^*}[U(E^*)].
    \]
This inequality holds for all $P_1 \in \mathcal{P}_1$. Therefore, the infimum occurs at the least favorable distribution:
  \begin{equation}
\label{eq:1}
\inf_{P_1 \in \mathcal{P}_1} \mathbb{E}_{P_1}[U(E^*)] = \mathbb{E}_{P_1^*}[U(E^*)].
\end{equation}

Finally, we show that no other candidate $E \in \mathcal{E}^\prime(\mathcal{P}_0)$ can achieve a higher worst-case growth rate.
Since $P_0^* \in \mathcal{P}_0$, $E$ must satisfy the validity constraint for this specific distribution:
  \[ \mathbb{E}_{P_0^*}[E] \le 1.\]
Furthermore, the worst-case growth of $E$ over the entire set $\mathcal{P}_1$ is bounded above by its growth against the specific distribution $P_1^*$:
  \[ \inf_{P_1 \in \mathcal{P}_1} \mathbb{E}_{P_1}[U(E)] \le \mathbb{E}_{P_1^*}[U(E)]. \]

Consider the simple hypothesis testing problem $P_0^*$ vs $P_1^*$. From our assumption,

$E^* = \argmax_{E \in \mathcal{E}^\prime(\{P_0\})}  \mathbb{E}_{P_1^*}[U(E)]$. Thus:
  \[ \mathbb{E}_{P_1^*}[U(E)] \le \mathbb{E}_{P_1^*}[U(E^*)]. \]

Combining these inequalities:
  \[
    \inf_{P_1 \in \mathcal{P}_1} \mathbb{E}_{P_1}[U(E)] \le \mathbb{E}_{P_1^*}[U(E)] \le \mathbb{E}_{P_1^*}[U(E^*)] = \inf_{P_1 \in \mathcal{P}_1} \mathbb{E}_{P_1}[U(E^*)].
    \]
Since this holds for any $E \in \mathcal{E}_b(\mathcal{P}_0)$, $E^*$ maximizes the worst-case growth rate, i.e.,
\begin{equation}
\label{eq:2}
\sup_{E \in \mathcal{E}_b(\mathcal{P}_0)} \inf_{P_1 \in \mathcal{P}_1} \mathbb{E}_{P_1}[U(E)] = \inf_{P_1 \in \mathcal{P}_1} \mathbb{E}_{P_1}[U(E^*)].
\end{equation}
Combining \eqref{eq:1} and \eqref{eq:2}, we obtain the desired result.
\end{proof}

\section{Counterexample with no LFD}
Consider the counterexample in \Cref{sec:discussion} with $E^*$ and $E^\prime$ defined in \eqref{q:e-star} and \eqref{eq:e-prime} respectively. Then, the next result shows that one can choose $\mu$ and $c$ such that $E^\prime$ has a strictly larger growth rate than that of $E^*$.
\begin{proposition}
\label{prop:counterexample}
For any $c>1$ and $\mu \in (0, 1/2)$ such that $ \mu + \frac{c-1}{\lambda^*}<1$, we have
    \begin{equation*}
    \mathbb{E}_{{Q}}[\log E^\prime] > \mathbb{E}_{{Q}}[\log E^*].
\end{equation*}
\end{proposition}
To construct a concrete counterexample, let us
take $\mu=0.25$ and $c=3$. Solving the first-order equation numerically yields $\lambda^*\approx 3.6$, and the condition $ \mu + \frac{c-1}{\lambda^*}<1$ is satisfied. In this case, the above result indicates that the ``optimize–then–constrain'' approach yields a strictly suboptimal solution. 
\begin{proof}[Proof of Proposition \ref{prop:counterexample}]
    \textbf{Step 1: Bounding the Multiplier $\gamma$.} \\
To ensure $E^*$ is a valid e-variable for the entire composite class $\mathcal{P}$, it must be valid for the Dirac measure $\delta_\mu \in \mathcal{P}$ (a point mass at $Z = \mu$). 
Evaluating $E^*$ under $\delta_\mu$ requires $\mathbb{E}_{\delta_\mu}[E^*] = E^*(\mu) \le 1$. Thus:
\begin{equation}
    \min\left(c, \frac{1}{\gamma}(1 + 0)\right) \le 1 \implies \frac{1}{\gamma} \le 1 \implies \gamma \ge 1.
\end{equation}
Because the log-growth objective is strictly decreasing in $\gamma$, the most competitive valid version of the projected e-variable sets $\gamma = 1$. Thus, we evaluate:
\begin{equation}
    E^*(Z) = \min\left(c, 1 + \lambda^*(Z - \mu) \right).
\end{equation}

\textbf{Step 2: Differentiating the Constrained Objective.} \\
We must prove that $\lambda_{new} \neq \lambda^*$. We do this by evaluating the derivative of the constrained objective $J_{con}(\lambda)$ at the unconstrained optimal point $\lambda^*$. 

Let $z_c(\lambda) = \mu + \frac{c-1}{\lambda}$ be the threshold where the truncation becomes active (i.e., $1 + \lambda(Z-\mu) = c$). The objective function splits into two regions:
\begin{equation}
    J_{con}(\lambda) = \int_0^{z_c(\lambda)} \log(1 + \lambda(z - \mu)) dQ(z) + \int_{z_c(\lambda)}^1 \log(c) dQ(z).
\end{equation}
By Leibniz's integral rule, the boundary terms generated by differentiating the limits of integration perfectly cancel out because the integrand is continuous at $z_c(\lambda)$ (specifically, $z_c(\lambda) = \mu + \frac{c-1}{\lambda}$ or $1 + \lambda(z_c(\lambda)-\mu) = c$). Thus, the derivative with respect to $\lambda$ is s
\begin{equation}
\label{eq:first-deriv}
    J_{con}'(\lambda) =\log c\cdot \frac{d}{d\lambda}z_c(\lambda)+ \int_0^{z_c(\lambda)} \frac{z - \mu}{1 + \lambda(z - \mu)} dQ(z)-\log c\cdot \frac{d}{d\lambda}z_c(\lambda)=\int_0^{z_c(\lambda)} \frac{z - \mu}{1 + \lambda(z - \mu)} dQ(z).
\end{equation}

\textbf{Step 3: Evaluating the Gradient at $\lambda^*$.} \\
By definition, $\lambda^*$ is the unique root of the unconstrained objective's derivative. Thus:
\begin{equation}
    \int_0^1 \frac{z - \mu}{1 + \lambda^*(z - \mu)} dQ(z) = 0.
\end{equation}
We can split this unconstrained integral at the truncation threshold $z_c(\lambda^*)$:
\begin{equation}
    \int_0^{z_c(\lambda^*)} \frac{z - \mu}{1 + \lambda^*(z - \mu)} dQ(z) + \int_{z_c(\lambda^*)}^1 \frac{z - \mu}{1 + \lambda^*(z - \mu)} dQ(z) = 0.
\end{equation}
Notice that the first term is exactly $J_{con}'(\lambda^*)$. Rearranging yields:
\begin{equation}
    J_{con}'(\lambda^*) = - \int_{z_c(\lambda^*)}^1 \frac{z - \mu}{1 + \lambda^*(z - \mu)} dQ(z).
\end{equation}

Assume the constraint $c$ is active, meaning $z_c(\lambda^*) < 1$. Because $c > 1$, we know $z_c(\lambda^*) = \mu + \frac{c-1}{\lambda^*} > \mu$. Therefore, over the entire domain of integration $(z_c(\lambda^*), 1]$, it holds that $z > \mu$. 
Consequently, the integrand $\frac{z - \mu}{1 + \lambda^*(z - \mu)}$ is strictly positive. Also, we assumed that $z_c(\lambda^*)=\mu + \frac{c-1}{\lambda^*}<1$ and hence $Q$ has positive support on this interval $(z_c(\lambda^*), 1]$, the integral is strictly positive, giving:
\begin{equation}
    J_{con}'(\lambda^*) < 0.
\end{equation}

\textbf{Step 4: Conclusion of Strict Dominance.} \\
The constrained objective $J_{con}(\lambda)$ is strictly concave in $\lambda$ because from \eqref{eq:first-deriv}, we have
\[J_{con}^{\prime\prime}(\lambda) =-\frac{(c-1)^2}{c\lambda^3}-\int_{0}^{z_c(\lambda)}\frac{1}{(\lambda+\frac{1}{z-\mu})^2}dz<0.\]
Because $J_{con}'(\lambda^*) < 0$, the function $J_{con}(\lambda)$ is strictly decreasing at $\lambda^*$. Therefore, the unique global maximizer $\lambda_{new}$ must lie strictly to the left of $\lambda^*$ (i.e., $\lambda_{new} < \lambda^*$). 

Since $\lambda^*$ is not the maximizer of the strictly concave function $J_{con}$, we rigorously conclude:
\begin{equation}
    J_{con}(\lambda_{new}) > J_{con}(\lambda^*).
\end{equation}
By definition, $J_{con}(\lambda_{new}) = \mathbb{E}_{Q}[\log E^\prime(Z)]$ and $J_{con}(\lambda^*) = \mathbb{E}_{Q}[\log E^*(Z)]$. Therefore:
\begin{equation}
    \mathbb{E}_{Q}[\log E^\prime(Z)] > \mathbb{E}_{Q}[\log E^*(Z)],
\end{equation}
concluding our proof.
\end{proof}

\end{document}